\documentclass[journal, onecolumn]{IEEEtran}
\usepackage{cite}

\ifCLASSINFOpdf
  \usepackage[pdftex]{graphicx}
\else
\fi
\usepackage{amsmath}
\usepackage{array}

\usepackage{stfloats}
\usepackage{amssymb}  
\usepackage{amsthm}  
\usepackage{hyperref}  
\usepackage{xcolor}  
\def\ds{\displaystyle}  
\DeclareMathOperator{\esssup}{ess\,sup}  
\newtheorem{theorem}{Theorem}
\newtheorem{corollary}{Corollary}
\newtheorem{remark}{Remark}
\newtheorem{algorithm}{Algorithm}

\begin{document}
%
\title{Quickest Change Detection with Cost-Constrained Experiment Design}
%
%
%

\author{Patrick Vincent N. Lubenia
        and Taposh Banerjee~\IEEEmembership{}
\thanks{The authors are with the University of Pittsburgh, Pittsburgh, PA 15260 USA (email: pnl8$@$pitt.edu, taposh.banerjee$@$pitt.edu). This paper was presented in part at the 61st Allerton Conference on Communication, Control, and Computing on September 17--19, 2025 at Urbana, Illinois.}
}

\maketitle


\begin{abstract}
In the classical quickest change detection problem, an observer performs a single experiment to monitor a stochastic process. The goal in the classical problem is to detect a change in the statistical properties of the process, with the minimum possible delay, subject to a constraint on the rate of false alarms. 
This paper considers the case where, at each observation time, the decision-maker must choose between multiple experiments with varying information qualities and costs. The change can be detected using any of the experiments. 
The goal here is to detect the change with the minimum delay, subject to constraints on the rate of false alarms and the fraction of time each experiment is performed before the time of change. The constraint on the fraction of time can be used to control the overall cost of using the system of experiments. An algorithm called the two-experiment cumulative sum (2E-CUSUM) algorithm is first proposed to solve the problem when there are only two experiments. 
The algorithm for the case of multiple experiments, starting with three experiments, is then designed iteratively using the 2E-CUSUM algorithm. Two key ideas used in the design 
are the scaling of undershoots and the truncation of tests. 
The multiple-experiment algorithm can be designed to satisfy the constraints and can achieve the delay performance of the experiment with the highest quality within a constant. The important concept of data efficiency, where the observer has the choice of not performing any experiment, is explored as well. 
\end{abstract}


%
\IEEEpeerreviewmaketitle


\section{Introduction}
%
%
%
%
\IEEEPARstart{T}{he} problem of quickest change detection (QCD) involves a stochastic process whose statistical properties change at some point in time and which an observer (or decision-maker) wishes to detect as soon as possible. QCD algorithms have been developed for statistical process control, sensor networks surveillance, computer network security, and network analysis, among others  \cite{veeravalli2013, poor2009, tartakovsky2014, tartakovsky2019}. These algorithms have been developed to optimize a metric on the detection delay, subject to a constraint on a metric on the rate of false alarms.  We provide an overview of the relevant literature in Section~\ref{sec:PriorWork}. 


The classical QCD problem involves observing outputs from one experiment that monitors the stochastic process under investigation. In this paper, we study QCD with an experiment design where there are multiple experiments that are able to detect the change in the process. Specifically, we assume that when a change occurs, no matter which experiment is performed, the change can be detected with that experiment. The decision-maker can choose only one of the experiments during an observation time, and there is an associated cost with each experiment. We also study the case of data-efficiency, where there is a constraint on the fraction of time any experiment is used. There is extensive literature on sequential hypothesis testing and QCD, particularly in the context of experiment design. We provide a brief overview of this literature in Section~\ref{sec:PriorWork}. To the best of our knowledge, the precise setup we study in this paper has not been investigated before. 

As an example of our problem setup, consider a battery-operated intruder or anomaly detection system (e.g., an autonomous vehicle) with a motion sensor and a video sensor, with the assumption that the latter sensor gives more accurate information but is costlier than the former. In this setup, one can perform three experiments: 1) activating and observing just the motion sensor, 2) using only the video sensor, and 3) using both sensors at the same time. At any given time, the observer can choose to perform only one of the three experiments. Budget and energy consumption constraints can limit the number of times an experiment is chosen, so the decision-maker has to carefully choose which experiment to perform, leading to a constraint in the fraction of time each sensor is used. In some situations, when the decision-maker believes that no change or anomaly can occur, it might decide not to employ any sensors, leading to energy efficiency or data efficiency. More generally, when there are a total of $\ell$ sensors, there are $2^{\ell}$ possible experiments one can perform, including the one where no experiment is chosen. 

In this paper, we develop an algorithm or design a policy that can sequentially detect the change as soon as it occurs, while also deciding which experiment to perform at each time step. The policy avoids false alarms and also satisfies various constraints on the fraction of time each experiment is performed before the change occurs. We also show that this algorithm is second-order optimal, i.e., optimal within a constant of the best possible algorithm. While we discussed experiments involving sensors as examples, the mathematical discussion in the paper applies to an abstract concept of an experiment qualified by its Kullback-Leibler divergence. 

We begin our mathematical discussion by first providing a problem setup for the case of $m$ experiments (Section~\ref{sec:prob}). Before providing a solution and analysis for this, we consider first the special cases of two- (Section~\ref{sec:two-expt}) and three-experiment (Section~\ref{sec:three-expt}) systems to understand the various issues involved. To minimize the worst-case average detection delay (a metric that we use for delay) for the two-experiment design case, we propose an algorithm that we call the Two-Experiment CUSUM (2E-CUSUM) algorithm. A detection threshold is used to satisfy the average run length to false alarm constraint, while a scaling factor and truncated test are used to satisfy what we call the Pre-change Observation Ratio, which is the fraction of time an experiment is performed. 

In sequential analysis, a statistical test can be truncated in various ways (random or deterministic). In this paper, we use a notion of random truncation to extend the two-experiment solution to multiple experiments. In our paper, a truncated test provides a limit on the number of times an experiment can be performed before reverting to a higher-quality experiment. For the three-experiment scenario, we extend the 2E-CUSUM algorithm and use its truncated version to define the 3E-CUSUM algorithm. We use this ability to define the 3E-CUSUM algorithm in terms of the 2E-CUSUM algorithm to define the algorithm to solve the $m$-experiment problem in an iterative manner (Section~\ref{sec:multiple-expt}): the 3E-CUSUM algorithm is defined using truncated 2E-CUSUM algorithm, the 4E-CUSUM algorithm defined using truncated 3E-CUSUM algorithm, and finally, $m$E-CUSUM defined in terms of a truncated version of the $(m-1)$E-CUSUM algorithm. 

We explore data-efficiency in Section~\ref{sec:data-efficient}. Similar to the non-data-efficient case, we first design a data-efficient algorithm where there are two experiments. We then extend the developed algorithm to the multiple-experiment scenario in an iterative manner.


\subsection{Prior Work}
\label{sec:PriorWork}
In the classical problem of QCD, there is only one experiment to be performed. Algorithms and optimality theories in this single-experiment setup are developed in \cite{shewhart1925, page1954, shiryaev1963, lorden1971, moustakides1986optimal, lai1998information, tartakovsky2005general, veeravalli2013, poor2009, tartakovsky2014, tartakovsky2019}. 

The problem of sequential hypothesis testing with experiment design has been studied in \cite{krishnamurthy2016, chernoff1959, bessler1960-2, albert1961, kiefer1963, lalley1986, keener1984, nitinawarat2013, naghshvar2013, nitinawarat2015, deshmukh2021, gurevich2019, hemo2020, vaidhiyan2018, tsopelakos2019, tsopelakos2020, tsopelakos2023}. A foundational idea in this type of analysis is that there are unknown parameters in the observation process, and the notion of the best experiment to perform depends on these unknown parameters. For example, in an intrusion detection problem, an unknown parameter could capture the true location of an intruder, and the optimal experiment would be to collect data near that location. The problem studied in our paper is a QCD problem, and we assume that the optimal experiment is known, i.e., there are no unknown parameters in the distributions of the observed data. 

The problem of QCD with experiment design has been studied in \cite{veeravalli2024, gopalan2021, dragalin1996, zhang2022, xu2021, chaudhuri2021, xu2023}. In some of these papers \cite{veeravalli2024, gopalan2021}, there are unknown parameters in the post-change distribution, and the optimal experiment to perform depends on this unknown parameter. Also, there is no direct cost constraint on the experiments. In other works \cite{dragalin1996, zhang2022, xu2021, chaudhuri2021, xu2023}, the problem is studied in a multi-stream or multi-channel setup, where there is a constraint on the number of streams one can observe at a time. Our problem setup differs from those studied in these papers because we assume that there is a cost for each experiment, and there is a notion of the best experiment that can be performed. In this sense, there is an asymmetry in our problem that is absent from the problems studied in these papers. 

The problem of QCD with experiment design, with cost constraints on the experiment, has been studied in \cite{banerjee2013, banerjee2012-2, banerjee2015, banerjee2013-2, banerjee2013-3, banerjee2015-2, banerjee2015-3} and \cite{ren2017}. The papers \cite{banerjee2013, banerjee2012-2, banerjee2015, banerjee2013-2, banerjee2013-3, banerjee2015-2, banerjee2015-3} study the problem where there is only one experiment, but there is a penalty associated with performing that experiment all the time or a reward for data-efficiency or energy-efficiency. This paper extends the works in these papers by allowing for multiple experiments. The nature of our results is similar in spirit to those in \cite{banerjee2013, banerjee2012-2, banerjee2015, banerjee2013-2, banerjee2013-3, banerjee2015-2, banerjee2015-3}. But we provide a comprehensive study while addressing several issues that are encountered in such an extension. 

The work in \cite{ren2017} extends the works in \cite{banerjee2013, banerjee2012-2, banerjee2015, banerjee2013-2, banerjee2013-3, banerjee2015-2, banerjee2015-3} by allowing for two experiments without resorting to data-efficiency. Our general algorithm for multiple experiments also applies to the two-experiment case. However, our two-experiment solution differs from that in \cite{ren2017} in several ways. First, in our algorithm, we begin processing data by using the experiment with higher quality and switch to the experiment with lower quality only when the data suggests we should do so. Second, our algorithm can be designed to satisfy every possible constraint value. Finally, our algorithm design is different and is amenable to extension to multiple experiments. We discuss our two-experiment solution in Section~\ref{sec:two-expt} and provide detailed justification for each of the design choices.

\section{Problem Formulation}
\label{sec:prob}

In this section, we formulate the problem of quickest change detection (QCD) with cost-constrained experiment design.

Consider a process whose statistical properties change at time $\nu$. At each time point, a decision-maker observes the process by performing one of $m$ experiments, with experiment $m$ having the highest information quality, experiment $m - 1$ the next-highest, and so on, and experiment 1 having the lowest quality (to be made precise below). Let $X_n^i$ be the observation at time $n$ using experiment $i$. The sequence of random variables $\{ X_n^i \}$ has a probability density function $f_0^i$ for $n < \nu$, and $f_1^i$ for $n \geq \nu$. Moreover, the sequence $\{ X_n^j \}$ has a higher information quality than $\{ X_n^i \}$ for $j > i$, as measured by their Kullback-Leibler divergence, i.e.,
\begin{equation} 
\label{eq:InformationQuality}
D(f_1^1 \mathrel{\Vert} f_0^1) \leq D(f_1^2 \mathrel{\Vert} f_0^2) \leq \dots \leq D(f_1^{m - 1} \mathrel{\Vert} f_0^{m - 1}) \leq D(f_1^m \mathrel{\Vert} f_0^m).
\end{equation}
As a consequence, in many practical situations, it seems reasonable to assume that the experiment $j$ is costlier to perform than experiment $i$, for $j > i$.

At each time $n \geq 0$, the decision-maker decides whether or not a change has occurred based on the information at time $n$. If the information available suggests no change has occurred, the decision-maker decides which experiment to perform at time $n + 1$. To formalize this, let $S_n^i$ (``switch") be the indicator random variable such that $S_n^i = 1$ if experiment $i$ is used for decision-making, and $S_n^i = 0$ otherwise. Also, let
$$ S_n = (S_n^1, \dots, S_n^m). $$ 
The decision vector $S_{n + 1}$ is a function of the information $I_n$
available at time $n \geq 1$:
\begin{equation} 
S_{n+1} = \phi_n(I_n), \; n \geq 0. 
\end{equation}
Here $I_n = (I_n^1, I_n^2, \dots, I_n^m)$ with 
$$ I_n^i = \left[ S_1^i, \dots, S_n^i, (X_1^i)^{(S_1^i)}, \dots, (X_n^i)^{(S_n^i)} \right], $$
and $I_0^i = \varnothing$. For $k = 1, \dots, n$, $(X_k^i)^{(S_k^i)}$ represent $X_k^i$ if $S_k^i = 1$. Otherwise, $X_k^i$ is not in $I_n^i$. 

At any time $n$, using the information $I_n$ available at that time, the decision-maker decides whether a change has occurred using a stopping time $\tau_{\textup{\fontsize{6 pt}{10 pt}\selectfont $m$E}}$. Therefore, a policy for the multiple-experiment QCD is $\Psi = \{ \tau_{\textup{\fontsize{6 pt}{10 pt}\selectfont $m$E}}, \phi_0, \dots, \phi_{\tau_{\text{$m$E}} - 1} \}$.

Each experiment has a different cost, with higher-quality experiments possibly being costlier than lower-quality experiments. Thus, we seek control policies that meet a constraint on the fraction of time each experiment is performed before the change point. We propose the following observation cost metric, which we call the Pre-change Observation Ratio ($\textup{POR}$), with $\textup{POR}_i$ being the POR for experiment $i$:
\begin{equation} 
\label{eq:POR_def}
\begin{split}
\textup{POR}_i (\Psi)
& = \limsup_{n \to \infty} \mathsf{E}_n \left[ \frac{1}{n} \sum_{k = 1}^{n - 1} S_k^i\right] \\
& = \limsup_{n \to \infty} \mathsf{E}_\infty \left[ \frac{1}{n} \sum_{k = 1}^{n - 1} S_k^i\right].
\end{split}
\end{equation}
Clearly, $\textup{POR}_i (\Psi) \leq 1$. If experiment $i$ is performed all the time in a policy, then $\textup{POR}_i (\Psi) = 1$. However, we avoid this case because this scenario will be too costly for the decision-maker. Hence, we consider the case where $\textup{POR}_i (\Psi) < 1$. We note that another way to define $\textup{POR}_i$ is to define it as
$$ \limsup_{n \to \infty} \mathsf{E}_\infty \left[ \frac{1}{n} \sum_{k = 1}^{n - 1} S_k^i \; \Bigg\vert \; \tau \geq n \right]. $$
However, based on the analysis done in \cite{banerjee2013} and \cite{banerjee2015}, it is clear that both expressions lead to similar analysis, with the unconditional metric leading to simpler analysis. The metric defined in \eqref{eq:POR_def} is interpreted as one where the metric is calculated when the observation process continues forever without ever stopping. 

For the mean time to false alarm, also called the average run length to false alarm (ARLFA), we consider the metric used in \cite{lorden1971} and \cite{pollak1985}, and constrain it to be a nonnegative number $\gamma$:
$$ \textup{ARLFA} (\Psi) = \mathsf{E}_\infty [\tau_{\textup{\fontsize{6 pt}{10 pt}\selectfont $m$E}}] \geq \gamma. $$

For the detection delay, we consider Lorden's criterion where we minimize the worst-case average detection delay (WADD) \cite{lorden1971}:
$$ \textup{WADD} (\Psi) = \sup_{\nu \geq 1} \esssup \mathsf{E}_\nu [(\tau - \nu + 1)^+ \mid I_{\nu - 1}]. $$

Our optimization problem is a minimax formulation for the multiple-experiment QCD problem:
\begin{align}
\begin{split}
\min_{\Psi} \quad & \textup{WADD} (\Psi) \\
\text{s.t. } \quad & \textup{ARLFA} (\Psi) \geq \gamma \\
& \textup{POR}_i (\Psi) \leq \beta_i \text{ for } i = 2, \dots, m
\end{split}
\label{eq:mE-qcd}
\end{align}
where $\gamma \geq 0$, $0 < \beta_m < 1$, and $0 \leq \beta_i < 1$ for $i = 2, \dots, m - 1$ are given constants. Note that $\textup{POR}_1 (\Psi) \leq \beta_1$ where $\ds \beta_1 = 1 - \sum_{i = 2}^m \beta_i$.

\begin{remark}
In Section \ref{sec:data-efficient}, we consider the case where $\ds \sum_{i = 1}^m \beta_i < 1$. In this instance, $\ds 1 - \sum_{i = 1}^m \beta_i$ is the fraction of time no experiment is performed. We call this the data-efficiency case of the multiple-experiment QCD problem.
\end{remark}

The challenge is to develop an asymptotically optimal algorithm that meets the constraints on ARLFA and POR. Our proposed solution or algorithm is given in Section \ref{sec:multiple-expt}, where we show that we can indeed design our algorithm to satisfy any set of given constraints on ARLFA and POR. In addition, the detection delay performance (WADD) of the algorithm is within a constant of the WADD of the cumulative sum (CUSUM) algorithm applied only to data from the highest quality experiment. However, in order to appreciate our algorithm and understand its analysis, we first study the special cases of two- and three-experiment designs.


\section{QCD with Two-Experiment Design}
\label{sec:two-expt}

In this section, we consider the special case of the two-experiment QCD problem and propose an algorithm to solve it. We then analyze the algorithm and show its asymptotic optimality. For simplicity of notation, we use variables $X$ and $Y$ to denote the observations from the two experiments. 

Suppose a decision-maker wants to detect the change in a stochastic process by performing one of two experiments at any given time: experiment $X$ or experiment $Y$, with experiment $Y$ having the higher information quality, i.e., Kullback-Leibler divergence (KL-divergence). If the pre- and post-change densities for $X$ are given by $f_0^X$ and $f_1^X$, respectively, and the corresponding densities for $Y$ are $f_0^Y$ and $f_1^Y$, then we have
\begin{equation}
D(f_1^X \mathrel{\Vert} f_0^X) \leq D(f_1^Y \mathrel{\Vert} f_0^Y).
\end{equation}

We wish to minimize detection delay subject to an ARLFA constraint and constraints on the fraction of time each experiment is performed. Let $\Psi$ be the policy for the two-experiment QCD. If $\beta_Y$ is the fraction of time experiment $Y$ is performed and $\textup{POR}_Y$ is the POR of experiment $Y$, then the two-experiment QCD problem is
\begin{align}
\begin{split}
\min_{\Psi} \quad & \textup{WADD} (\Psi) \\
\text{s.t. } \quad & \textup{ARLFA} (\Psi) \geq \gamma \\
& \textup{POR}_Y (\Psi) \leq \beta_Y
\end{split}
\label{eq:2e-qcd}
\end{align}
where $\gamma \geq 0$ and $0 < \beta_Y < 1$ are given constants. Note that $\textup{POR}_X (\Psi) \leq \beta_X$ where $\beta_X = 1 - \beta_Y$.


\subsection{The \textup{2E-CUSUM} Algorithm}
\label{sec:2e-cusum}

The classical Cumulative Sum (CUSUM) algorithm, proposed in \cite{page1954}, can be described as follows:

\begin{algorithm}[CUSUM: $\Psi_{\textup{C}}^X (A) = \Psi_{\textup{C}}^X$]
\label{alg:cusum}

Fix $A > 0$. The \textup{CUSUM} statistic $\{ C_n^X \}$ for experiment $X$ is defined as follows:

\begin{enumerate}
\item Statistic Calculation: We start with $C_0^X = 0$ and update $\{ C_n^X \}$ as
$$ C_{n + 1}^X = \max \{ C_n^X + \ell_X (X_{n + 1}), 0 \}, $$
where $\ell_X (x) = \log \frac{f_1^X (x)}{f_0^X (x)}$.

\item Stopping Rule: Stop at
$$ \tau_{\textup{\fontsize{6 pt}{10 pt}\selectfont C}}^X =\tau_{\textup{\fontsize{6 pt}{10 pt}\selectfont C}}^X(A)= \inf \{ n \geq 1: C_n^X > A \}. $$
\end{enumerate}    
\end{algorithm}

\begin{remark}
\label{rem:CUSUM4Y}
The corresponding variables for experiment $Y$ are $\Psi_{\textup{C}}^Y$, $C_n^Y$, $\ell_Y$, $\tau_{\textup{\fontsize{6 pt}{10 pt}\selectfont C}}^Y$, and $\tau_{\textup{\fontsize{6 pt}{10 pt}\selectfont C}}^Y(A)$. 
\end{remark}

\medskip

We build on this algorithm and allow negative values for the statistic to develop what we call the Two-Experiment CUSUM (2E-CUSUM) algorithm, which minimizes WADD for the two-experiment QCD problem. This approach of using negative values is inspired by the work in \cite{banerjee2013}. The algorithm proposed below needs two important parameters: a scale factor $a_Y$ and the limit $N_X$ for a truncated test. We first discuss the algorithm and then the motivation behind its design, including comparisons with a similar algorithm developed in \cite{ren2017}.  

\begin{algorithm}[2E-CUSUM: $\Psi_{\textup{2E}} (A, a_Y, N_X)$]
\label{alg:2e-cusum}

Fix $A > 0$, $a_Y > 0$, and $N_X \geq 0$. The \textup{2E-CUSUM} statistic $\{ D_n \}$ (``dual-experiment") is defined as follows:

\begin{enumerate}
\item Sampling Control: For $n \geq 0$, the experiment selection control is
$$ S_{n + 1} =
\left\{
\begin{array}{ll}
1 & \text{if } D_n \geq 0 \\
0 & \text{if } D_n < 0.
\end{array}
\right. $$
In other words, at time $n + 1$, the decision-maker performs experiment $Y$ when the statistic $D_n$ is nonnegative, and experiment $X$ when $D_n$ is negative.

\item Statistic Calculation: This is done in several steps:

    \begin{enumerate}
    \item We start with $D_0 = 0$ and update $\{ D_n \}$ using the observations $\{ Y_n \}$ from the higher quality experiment using
    $$ D_{n + 1}  = D_n + \ell_Y (Y_{n + 1}), $$
    where $\ell_Y (y) = \log \frac{f_1^Y (y)}{f_0^Y (y)}$, until 
    $$ \sigma := \sigma^Y (A) = \min \{ n \geq 1: D_n \notin (0, A) \}. $$
    If $D_\sigma > A$, we stop and declare the change.

    \item If $D_\sigma < 0$, we start using observations $\{ X_n \}$ from the lower quality experiment: We scale the undershoot $D_\sigma$ by a factor of $a_Y$, and use $a_Y D_\sigma$ as the new zero level for the \textup{CUSUM} for $X$ and update $D_n$ as
    $$ D_{n + 1} = \max \{ D_n + \ell_X (X_{n + 1}), \; a_Y D_\sigma \}. $$
    This continues until either $D_n$ goes above $0$ or $N_X$ number of observations of $X$ are used. At this time, we reset $D_n$ to $0$ and continue as in step $\textup{a)}$.  A value of $N_X \in (\ell, \ell + 1)$ for integer $\ell \geq 0$ means $N_X = \ell$ with probability $\ell +1 - N_X$ and $N_X = \ell + 1$ with probability $N_X-\ell$.  
    \end{enumerate}

\item Stopping Rule: Stop at the stopping time
$$ \tau_{\textup{\fontsize{6 pt}{10 pt}\selectfont 2E}} = \inf \{ n \geq 1: D_n > A \}. $$
\end{enumerate}    
\end{algorithm}

A typical evolution of the 2E-CUSUM algorithm is illustrated in Figure \ref{fig:2e-cusum}. The achievable $\textup{POR}_Y$ values as a function of the scale factor $a_Y$ and truncation $N_X$, as shown in Figure \ref{fig:por-2e-sim}, show that $\textup{POR}_Y$ decreases rapidly as a function of $a_Y$. This behavior exists because increasing the scale factor $a_Y$ is equivalent to increasing the threshold for the CUSUM algorithm for $X$ observations. Since  $\textup{POR}_Y$ value is calculated under the pre-change hypothesis, the drift of the CUSUM for $X$ is negative, and the time to reach the level $0$ is exponential in $|a_Y D_\sigma|$. The $\textup{POR}_Y$ values were calculated using Monte Carlo simulations. 

\begin{figure}[!t]
\centering
\includegraphics[width = 0.5\linewidth]{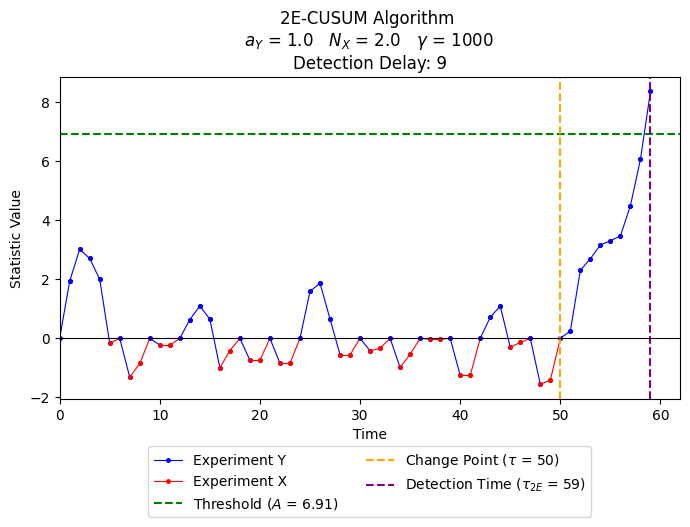}
\caption{An evolution of the statistic $D_n$ of the 2E-CUSUM algorithm for $f_0^X = f_0^Y = \mathcal{N}(0, 1)$, $f_1^X = \mathcal{N}(0.75, 1)$, $f_1^Y = \mathcal{N}(1, 1)$, $\nu = 50$, and $\gamma = 1000$, with $a_Y = 1.0$ and $N_X = 2.0$. The decision-maker performs experiment $Y$ whenever $D_n \geq 0$, and experiment $X$ when $D_n < 0$.}
\label{fig:2e-cusum}
\end{figure}

\begin{figure}[!t]
\centering
\includegraphics[width = 0.5\linewidth]{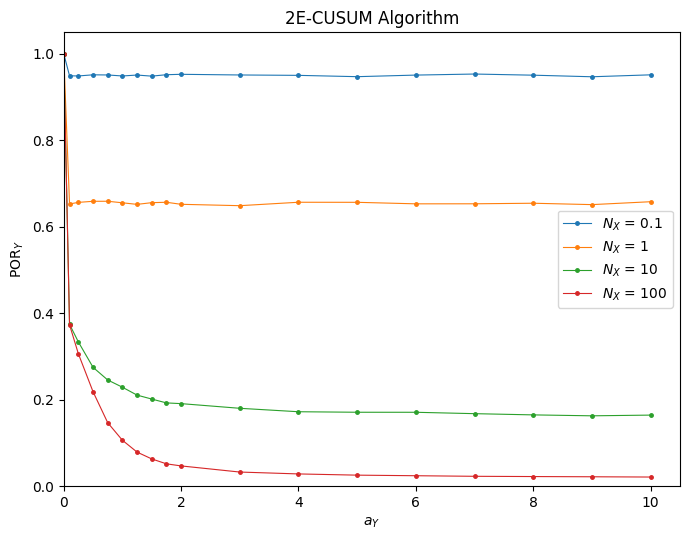}
\caption{$\textup{POR}_Y$ of the 2E-CUSUM algorithm as a function of the scaling factor $a_Y$ for $f_0^X = f_0^Y = \mathcal{N}(0, 1)$, $f_1^X = \mathcal{N}(0.75, 1)$, and $f_1^Y = \mathcal{N}(1, 1)$ with truncation $N_X \in \{ 0.1, 1, 10, 100 \}$. $\textup{POR}_Y$ exponentially decreases as a function of $a_Y$.}
\label{fig:por-2e-sim}
\end{figure}

\begin{remark}
We now comment on the design choices in the \textup{2E-CUSUM} algorithm. 

\begin{itemize}
\item \textup{Special case}: When $N_X = 0$, the \textup{2E-CUSUM} algorithm reduces to the \textup{CUSUM} algorithm for experiment $Y$ (and experiment $X$ is never performed).

\item \textup{Achieving high $\textup{POR}_Y$ values}: Unlike the algorithm proposed in \cite{ren2017}, the \textup{2E-CUSUM} algorithm can be designed to achieve any $\textup{POR}_Y$ value between 0 and 1. This is important in applications where a low value of delay is important, but experiment design is also desired. In these cases, one may choose a $\textup{POR}_Y$ value between 0.5 and 1.0 without significantly affecting the delay. In short, the best trade-off is often achievable for higher $\textup{POR}_Y$ values \cite{banerjee2013}.

\item \textup{Starting the decision process with the $Y$ experiment}: \textup{2E-CUSUM} starts by choosing experiment $Y$ first because we wish to use the better experiment to detect the change quickly if it occurs early. If the change occurs at a later time, then POR constraints imply that it will not matter which experiment is performed first. This is in contrast with the algorithm in \cite{ren2017}, which starts with experiment $X$.

\item \textup{Extension to multiple experiments and data-efficiency}: As compared to the work in \cite{ren2017}, we extend our algorithm to allow for multiple experiments and also allow for data efficiency. As discussed earlier in the paper, in a data-efficient setting, no experiment is performed for a positive fraction of time before the time of change. 
\end{itemize}
\end{remark}


\subsection{Analysis and Optimality of the \textup{2E-CUSUM} Algorithm}
\label{sec:2e-cusum-analysis}

In this section, we analyze the 2E-CUSUM algorithm and prove its optimality. We define some new notations to be used in the theorem. Recall that we defined in Algorithm \ref{alg:2e-cusum}, $\sigma^Y (A) = \min \{ n \geq 1: D_n \notin (0, A) \}$, i.e., the first time the statistic $D_n$ either falls below 0 or crosses $A$ from below. We let $\sigma^Y (\infty) = \sigma^Y$, i.e., the first time the statistic falls below 0. We also define (recall Algorithm~\ref{alg:cusum})
$$ \tau_{\textup{\fontsize{6 pt}{10 pt}\selectfont C}}^X (A, N_X) = \min \{ \tau_{\textup{\fontsize{6 pt}{10 pt}\selectfont C}}^X(A), N_X\}. $$ 
Also, recall from the definition of the 2E-CUSUM algorithm that $D_\sigma$ is the undershoot of the CUSUM statistic for experiment $Y$ when it falls below zero for the first time. Figure \ref{fig:por-ds} illustrates these notations.

\begin{figure}[!t]
\centering
\includegraphics[width = 0.5\linewidth]{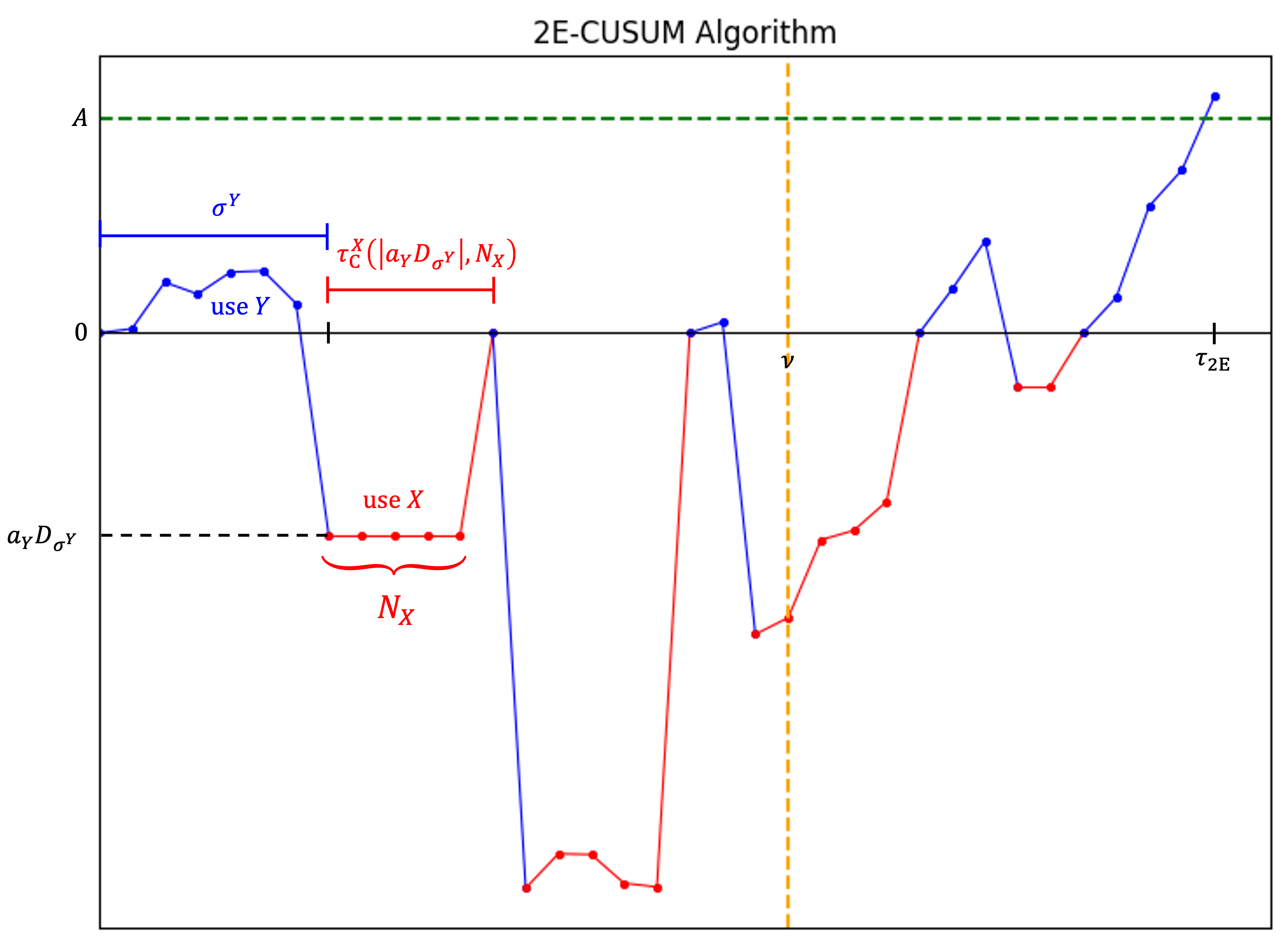}
\caption{The 2E-CUSUM Algorithm. Experiment $Y$ is performed first. When the statistic $D_n$ falls below 0, the undershoot is scaled by a factor of $a_Y$, and experiment $X$ is performed at most $N_X$ times. When the statistic crosses 0 from below, experiment $Y$ is performed again. Using pre-change distributions of $\{ X_n \}$ and $\{ Y_n \}$, $\sigma^Y$ is the length of time experiment $Y$ is performed while $\tau_{\textup{\fontsize{6 pt}{10 pt}\selectfont C}}^X (\vert a_Y D_{\sigma^Y} \vert, N_X)$ is the length of time experiment $X$ is performed before change occurs at time $\nu$. The algorithm stops and declares change at time $\tau_{\textup{\fontsize{6 pt}{10 pt}\selectfont 2E}}$ when the statistic crosses threshold $A$.}
\label{fig:por-ds}
\end{figure}

\begin{theorem}
\label{thm:2e}

For the \textup{2E-CUSUM} algorithm, let
$$ 0 < D(f_1^X \mathrel{\Vert} f_0^X) \leq D(f_1^Y \mathrel{\Vert} f_0^Y) < \infty. $$

\begin{enumerate}
\item For any fixed $a_Y > 0$ and $N_X > 0$, with $A = \log \gamma$,
$$ \textup{ARLFA} (\Psi_{\textup{2E}}) \geq \gamma. $$

\item For fixed values of $A$, $a_Y > 0$, and $N_X > 0$,
\begin{align*}
& \textup{POR}_Y (\Psi_{\textup{2E}} (A, a_Y, N_X)) = \frac{\mathsf{E}_\infty [\sigma^Y]}{\mathsf{E}_\infty [\sigma^Y] + \mathsf{E}_\infty [\tau_{\textup{\fontsize{6 pt}{10 pt}\selectfont C}}^X (\vert a_Y D_\sigma \vert, N_X)]},
\end{align*}
where we recall that $D_\sigma = D_{\sigma^Y}$ is the undershoot. Thus, for any choice of $A$, there exist $a_Y^* > 0$ and $N_X^*>0$ to meet any given POR constraint of $\beta_Y \in (0, 1]$ with equality:
$$ \textup{POR}_Y (\Psi_{\textup{2E}} (A, a_Y^*, N_X^*)) = \beta_Y. $$

\item For fixed values of $a_Y > 0$ and $N_X > 0$, and for each $A$, we have
$$ \textup{WADD} (\tau_{\textup{\fontsize{6 pt}{10 pt}\selectfont 2E}}) = \mathsf{E}_1 [\tau_{\textup{\fontsize{6 pt}{10 pt}\selectfont 2E}}] + N_X. $$
Consequently, we have 
$$ \textup{WADD} (\tau_{\textup{\fontsize{6 pt}{10 pt}\selectfont 2E}}) = \textup{WADD} (\Psi_{\textup{C}}^Y) + K_2, $$
where $\Psi_{\textup{C}}^Y$ is the \textup{CUSUM} algorithm for experiment $Y$ (see Remark~\ref{rem:CUSUM4Y}), and $K_2$ is a constant that is not a function of $A$. 
\end{enumerate}
\end{theorem}

\begin{proof}
\begin{enumerate}
\item It is well known that
$$ \mathsf{E}_\infty [\tau_{\textup{\fontsize{6 pt}{10 pt}\selectfont C}}^Y (A)] \geq e^A, $$
i.e., the expected length of time it takes for the CUSUM statistic (for $X$, $Y$, or any other experiment) to cross the threshold $A$, under the pre-change model, is lower-bounded by $e^A$ \cite{lorden1971}. Because of the independence of observations, the sojourns of the statistic $D_n$ below zero, and the time spent conducting the experiment $X$, only add to the time it takes to hit the threshold $A$. Thus, 
$$ \mathsf{E}_\infty [\tau_{\textup{\fontsize{6 pt}{10 pt}\selectfont 2E}}] \geq \mathsf{E}_\infty [\tau_{\textup{\fontsize{6 pt}{10 pt}\selectfont 2E}} (S_n = 1 \;\; \forall n)] = \mathsf{E}_\infty [\tau_{\textup{\fontsize{6 pt}{10 pt}\selectfont C}}^Y (A)] \geq e^A. $$
Setting $A = \log \gamma$, we get
$$ \mathsf{E}_\infty [\tau_{\textup{\fontsize{6 pt}{10 pt}\selectfont 2E}}] \geq e^{\log \gamma} = \gamma. $$

\item This result follows from the Renewal Reward Theorem because every time the statistic $D_n$ of the 2E-CUSUM algorithm goes above zero from below, it is reset to zero. This creates a renewal cycle with the reward being the time spent by statistic $D_n$ above $0$. Under the pre-change distribution, the expected number of observations of experiment $Y$ is $\mathsf{E}_\infty [\sigma^Y]$. Now, at time $\sigma^Y$, the new zero-level for the CUSUM for experiment $X$ is $a_Y D_\sigma$. Hence, the expected number of observations used of experiment $X$ before it crosses zero from below is $\mathsf{E}_\infty [\tau_{\textup{\fontsize{6 pt}{10 pt}\selectfont C}}^X (\vert a_Y D_\sigma \vert, N_X)]$. Therefore,
$$ \textup{POR}_Y (\tau_{\textup{\fontsize{6 pt}{10 pt}\selectfont 2E}}) = \frac{\mathsf{E}_\infty [\sigma^Y]}{\mathsf{E}_\infty [\sigma^Y] + \mathsf{E}_\infty [\tau_{\textup{\fontsize{6 pt}{10 pt}\selectfont C}}^X (\vert a_Y D_\sigma \vert, N_X)]}. $$
Moreover,
$$ \textup{POR}_X (\tau_{\textup{\fontsize{6 pt}{10 pt}\selectfont 2E}}) = 1 - \textup{POR}_Y (\tau_{\textup{\fontsize{6 pt}{10 pt}\selectfont 2E}}). $$
Finally, note that as $N_X \to 0$,
$$ \mathsf{E}_\infty [\tau_{\textup{\fontsize{6 pt}{10 pt}\selectfont C}}^X (\vert a_Y D_\sigma \vert, N_X)] \to 0 \quad \text{and} \quad \textup{POR}_Y \to 1. $$
Also, as $N_X \to \infty$ and $a_Y \to \infty$,
$$ \mathsf{E}_\infty [\tau_{\textup{\fontsize{6 pt}{10 pt}\selectfont C}}^X (\vert a_Y D_\sigma \vert, N_X)] \to \infty \quad \text{and} \quad \textup{POR}_Y \to 0. $$

\item Suppose the change occurs at time $\nu = 2$. The essential supremum operation in the WADD delay metric would cause the observation at time 1 to be such that the undershoot $D_\sigma \to \infty$. Then the only way for the statistic $D_n$ to go above zero would be to exhaust all the sample values $N_X$ designated for $X$. Similar situations when $\nu = 3, 4, \dots$ would give us
\begin{align*}
\textup{WADD} (\tau_{\textup{\fontsize{6 pt}{10 pt}\selectfont 2E}})
& = \mathsf{E}_1 [\tau_{\textup{\fontsize{6 pt}{10 pt}\selectfont 2E}}] + N_X \\
& = (\mathsf{E}_1 [\tau_{\textup{\fontsize{6 pt}{10 pt}\selectfont C}}^Y] + K) + N_X \\
& = \textup{WADD} (\Psi_{\textup{C}}^Y) + K + N_X
\end{align*}
where $K$ is a constant, not a function of the threshold $A$. The second equality follows because after the change occurs, and after a finite number of bounded excursions below $0$, the statistic $D_n$ would never go below 0 \cite{banerjee2013}. 
\end{enumerate}
\end{proof}

The above theorem implies strong optimality properties of the 2E-CUSUM algorithm. 

\begin{corollary}
Under the conditions of the theorem, suppose we set $A = \log \gamma$ to satisfy 
$$ \textup{ARLFA} (\Psi_{\textup{2E}}) \geq \gamma. $$
Also, we set $a_Y$ and $N_X$ to satisfy
$$ \textup{POR}_Y (\Psi_{\textup{2E}}) \leq \beta_Y. $$
Then as $\gamma \to \infty$,
$$ \textup{WADD} (\Psi_{\textup{2E}} (A, a_Y, N_X)) \sim \textup{WADD} (\Psi_{\textup{C}}) \sim \frac{\log \gamma}{D(f_1^Y \mathrel{\Vert} f_0^Y)}. $$
This shows that the \textup{2E-CUSUM} algorithm is asymptotically optimal for the two-experiment QCD problem stated in \eqref{eq:mE-qcd}.
\end{corollary}

\begin{proof}
The asymptotic performance follows from the performance of the CUSUM algorithm and the theorem statement. For the statement on asymptotic optimality, we note that if experiment $X$ is used for a fixed and positive fraction of time after the change, and all the way to the stopping, then the information number accumulated over time will be a convex combination of the KL-divergence for $X$ and the KL-divergence for $Y$. According to the general theory of change detection developed in \cite{lai1998information}, the delay performance will then be inversely proportional to this convex combination, resulting in a higher delay. Thus, any alternative solution to this problem must ultimately utilize $Y$, resulting in the delay we obtain with the 2E-CUSUM algorithm. 
\end{proof}


\subsection{Performance Plots for the 2E-CUSUM Algorithm}

In this subsection, we provide delay-false alarm trade-off curves for the 2E-CUSUM algorithm computed using Monte Carlo simulations. For ARLFA, each simulation of an algorithm is run using only pre-change distributions and, once the algorithm ends, the termination time is recorded. ARLFA is the average of the simulated termination times. WADD is computed in the same manner but using only post-change distributions. To get the value of the WADD metric, we add $N_X$ to the computed average termination times. These values are reported in Figure~\ref{fig:ds-por-compare} for different choices of the $\text{POR}_Y$ constraint. As can be seen from the figure, as the $\text{POR}_Y$ value goes to $1$, the performance of the 2E-CUSUM algorithm approaches that of the CUSUM algorithm using experiment $Y$. 

\begin{figure}[!t]
\centering
\includegraphics[width = 0.5\linewidth]{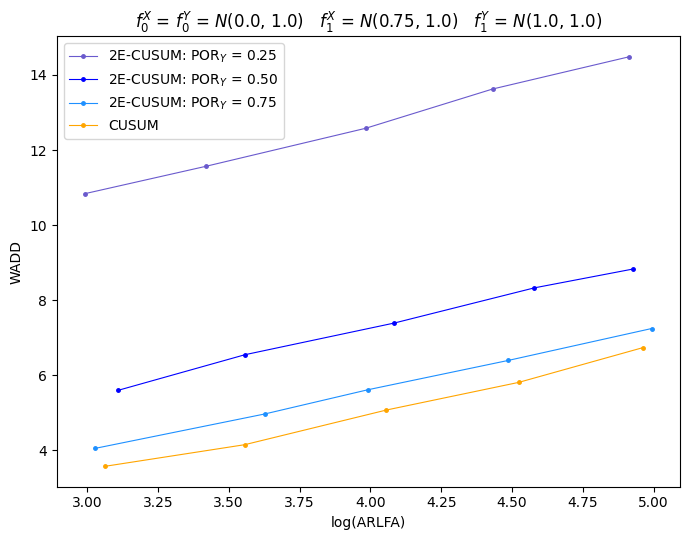}
\caption{WADD vs $\log (\textup{ARLFA})$ graphs of the 2E-CUSUM algorithm for $f_0^X = f_0^Y = \mathcal{N}(0, 1)$, $f_1^X = \mathcal{N}(0.75, 1)$, and $f_1^Y = \mathcal{N}(1, 1)$ and different $\textup{POR}_Y$ levels. The 2E-CUSUM algorithm performs better at a higher $\textup{POR}_Y$ level.}
\label{fig:ds-por-compare}
\end{figure}

We also compare the performance of our algorithm with an alternative way to solve the two-experiment problem, which we call the Random Switch Scheme (RSS). To evaluate the performance of the 2E-CUSUM algorithm against that of the CUSUM algorithm and the RSS algorithm, we plot their WADD vs $\log (\textup{ARLFA})$ graphs for different values of detection threshold $A$. The RSS algorithm works as follows. The statistic $D_n$ starts at 0, and $D_n$ evolves according to CUSUM using observations from either experiment $Y$ or experiment $X$. The decision-maker performs experiment $Y$ first. After updating $D_n$, a (possibly biased) coin is tossed: if it lands on heads, experiment $Y$ is done next; otherwise, experiment $X$ is performed. The process continues until $D_n \geq A$, at which point the algorithm stops and the decision-maker declares that a change has occurred. Figure \ref{fig:trade-off-ds} shows the WADD vs $\log (\textup{ARLFA})$ graphs of the 2E-CUSUM, CUSUM, and RSS algorithms. To allow proper comparison of the 2E-CUSUM and RSS algorithms, we set their parameters to correspond to the same POR levels of $\textup{POR}_Y = \textup{POR}_X = 0.50$. Among the three algorithms, since the CUSUM algorithm uses more informative observations all the time, it has the lowest detection delay. We also observe that the 2E-CUSUM algorithm outperforms the RSS algorithm, as evidenced by the former's lower WADD.

\begin{figure}[!t]
\centering
\includegraphics[width = 0.5\linewidth]{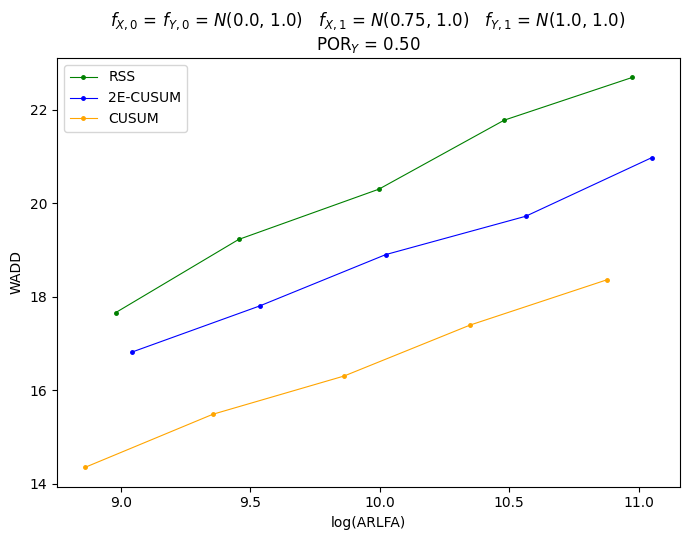}
\caption{WADD vs $\log (\textup{ARLFA})$ graphs of the 2E-CUSUM, CUSUM, and RSS algorithms for $f_0^X = f_0^Y = \mathcal{N}(0, 1)$, $f_1^X = \mathcal{N}(0.75, 1)$, and $f_1^Y = \mathcal{N}(1, 1)$. For the 2E-CUSUM algorithm, we set the parameters at $a_Y = 1.0$ and $N_X = 2.0$, corresponding to $\textup{POR}_Y = \textup{POR}_X = 0.50$. For the RSS algorithm, we set the probability of performing experiment $Y$ at $p_Y = 0.50$. The CUSUM algorithm has the lowest detection delay. The 2E-CUSUM algorithm outperforms the RSS algorithm.}
\label{fig:trade-off-ds}
\end{figure}


\subsection{Truncated \textup{2E-CUSUM} Algorithm}
\label{sec:trunc-2e}

In Section~\ref{sec:three-expt}, we will discuss our solution to QCD with the design of three experiments. There, we will design the three-experiment algorithm using the algorithm for two experiments. In Section~\ref{sec:multiple-expt}, we will show that the algorithm for QCD with four experiments can be defined using the algorithm for three experiments. Finally, the algorithm for $m$ experiments is defined using the algorithm for $m-1$ experiments. To define the algorithms in this iterative manner, we need and use a notion of a truncated test. There are many ways to define truncation, including hard truncation at a fixed time. The truncation we define here, however, is a softer version. In this subsection, we define a truncated version of the 2E-CUSUM algorithm. 

Observe that in the 2E-CUSUM algorithm, we do not limit the number of consecutive observations of experiment $Y$, but we limit that of experiment $X$ by $N_X$. We define now the truncated 2E-CUSUM algorithm as the 2E-CUSUM algorithm wherein we also limit the number of consecutive observations of experiment $Y$ by $N_Y$. While performing experiment $Y$, if the statistic crosses the detection threshold or $N_Y$ observations are reached, we stop and declare a change.

\begin{algorithm}[Truncated 2E-CUSUM: $\Psi_{\textup{2E}} (A, a_Y, N_X, N_Y)$]
\label{alg:truncated_2e-cusum}

Fix $A > 0$, $a_Y > 0$, $N_X \geq 0$, and $N_Y \geq 0$. The truncated \textup{2E-CUSUM} statistic $\{ D_n \}$ is defined as follows:

\begin{enumerate}
\item Sampling Control: For $n \geq 0$, the experiment selection control is
$$ S_{n + 1} =
\left\{
\begin{array}{ll}
1 & \text{if } D_n \geq 0 \\
0 & \text{if } D_n < 0
\end{array}
\right. $$
In other words, at time $n + 1$, the decision-maker performs experiment $Y$ when the statistic $D_n$ is nonnegative, and experiment $X$ if $D_n$ is negative.

\item Statistic Calculation: This is done in several steps:

    \begin{enumerate}
    \item We start with $D_0 = 0$ and update $\{ D_n \}$ using the observations $\{ Y_n \}$ from the higher quality experiment using
    $$ D_{n + 1}  = D_n + \ell_Y (Y_{n + 1}), $$
    where $\ell_Y (y) = \log \frac{f_1^Y (y)}{f_0^Y (y)}$, until either time
    $$ \sigma := \sigma^Y(A) = \min \{ n \geq 1: D_n \notin (0, A) \} $$
    or $N_Y$ number of observations of $Y$ are used. If $D_\sigma > A$ or $N_Y$ has been reached, we stop and declare the change.
    
    \item If $D_\sigma < 0$, we start using observations $\{ X_n \}$ from the lower quality experiment: we scale the  undershoot $D_\sigma$ by a factor of $a_Y$, and use $a_Y D_\sigma$ as the new zero level for the \textup{CUSUM} for $X$ and update $D_n$ as
    $$ D_{n + 1} = \max \{ D_n + \ell_X (X_{n + 1}), \;  a_Y D_\sigma \}. $$
    This continues until either $D_n$ goes above $0$ or $N_X$ number of observations of $X$ are used. At this time, we reset $D_n$ to $0$ and continue as in step $\textup{a)}$.
    \end{enumerate}

\item Stopping Rule: Stop at the stopping time
$$ \tau_{\textup{\fontsize{6 pt}{10 pt}\selectfont 2E}}(A, a_Y, N_X, N_Y) = \inf \{ n \geq 1: D_n > A \}. $$
\end{enumerate}    
\end{algorithm}

We note that instead of using the stopping rule notation $\tau_{\textup{\fontsize{6 pt}{10 pt}\selectfont 2E}}$ that we used for the 2E-CUSUM algorithm, we use the notation
$$
\tau_{\textup{\fontsize{6 pt}{10 pt}\selectfont 2E}}(A, a_Y, N_X, N_Y)
$$
to denote the stopping rule for the truncated 2E-CUSUM algorithm. In this new notation, note that
$$
\tau_{\textup{\fontsize{6 pt}{10 pt}\selectfont 2E}} = \tau_{\textup{\fontsize{6 pt}{10 pt}\selectfont 2E}}(A, a_Y, N_X, \infty). 
$$
It is clear from the description that 
$$ \tau_{\textup{\fontsize{6 pt}{10 pt}\selectfont 2E}}(A, a_Y, N_X, N_Y) \leq N_Y + N_Y N_X. $$


\vspace{1cm}
\section{QCD with Three-Experiment Design}
\label{sec:three-expt}

In this section, we consider the case where the decision-maker can choose between three experiments $X$, $Y$, or $Z$, where experiment $Z$ has the highest information quality while experiment $X$ has the lowest, i.e.,
\begin{equation}
D(f_1^X \mathrel{\Vert} f_0^X) \leq D(f_1^Y \mathrel{\Vert} f_0^Y) \leq D(f_1^Z \mathrel{\Vert} f_0^Z). 
\end{equation}
We first formulate the three-experiment QCD problem. We then present the Three-Experiment CUSUM (3E-CUSUM) algorithm, which solves the problem and is an extension of the 2E-CUSUM algorithm. Finally, the mathematical analysis of the algorithm and its optimality are presented.

Motivated by the design and optimality of the 2E-CUSUM algorithm, the algorithm for a three-experiment QCD problem will start with the CUSUM statistic for experiment $Z$, and when the statistic goes below zero, it will switch to the next costly experiment $Y$. At this point, we need to address, among other issues, how to switch to $X$, and more importantly, how to truncate the use of $Y$ and $X$ precisely to make sure that the POR constraints can be met with equality. At the same time, we need to make sure that WADD for the stopping rule is finite. The goal of this section is to address these issues in a precise manner.

The three-experiment QCD problem is similar to the two-experiment case, except we impose an additional constraint on the fraction of time experiment $Y$ is performed. By extending the definitions of the variables defined in the previous section, our optimization problem is a minimax formulation for the three-experiment QCD problem:
\begin{align}
\begin{split}
\min_{\Psi} \quad & \textup{WADD} (\Psi), \\
\text{s.t. } \quad & \textup{ARLFA} (\Psi) \geq \gamma, \\
& \textup{POR}_Z (\Psi) \leq \beta_Z, \\
& \textup{POR}_Y (\Psi) \leq \beta_Y,
\end{split}
\end{align}
where $\gamma \geq 0$, $0 < \beta_Z < 1$, and $0 \leq \beta_Y < 1$ are given constants. Note that $\textup{POR}_X (\Psi) \leq \beta_X$ where $\beta_X = 1 - \beta_Y - \beta_Z$.


\subsection{The \textup{3E-CUSUM} Algorithm}
\label{sec:3e-cusum}

We now propose the 3E-CUSUM algorithm. Recall from Section~\ref{sec:trunc-2e} that we use the notation 
$\Psi_{\textup{2E}} (A, a_Y, N_X, N_Y)$ to denote the truncated 2E-CUSUM algorithm with corresponding stopping rule defined as $\tau_{\textup{\fontsize{6 pt}{10 pt}\selectfont 2E}}(A, a_Y, N_X, N_Y)$. 

We use the important notation of
$$ \Psi_{\textup{2E}} (b,  A, a_Y, N_X, N_Y) $$
which is the truncated 2E-CUSUM starting at $b$. Here, the switch from $Y$ to $X$ occurs at $b$ (instead of zero), and the stopping threshold is $A$.  We also use 
$$ \tau_{\textup{\fontsize{6 pt}{10 pt}\selectfont 2E}}(b, A, a_Y, N_X, N_Y) $$
to denote the stopping rule of $\Psi_{\textup{2E}} (b, A, a_Y, N_X, N_Y)$.


\begin{algorithm}[3E-CUSUM: $\Psi_{\textup{3E}} (A, a_Y, a_Z, N_X, N_Y)$]
\label{alg:3e-cusum}

Fix $A > 0$, $a_Y > 0$, $a_Z > 0$, $N_X \geq 0$, and $N_Y \geq 0$. The \textup{3E-CUSUM} statistic $\{ D_n \}$ is defined as follows:

\begin{enumerate}
\item We start with $D_0 = 0$ and update $\{ D_n \}$ using the observations $\{ Z_n \}$ from the highest quality experiment using
$$ D_{n + 1}  = D_n + \ell_Z (Z_{n + 1}), $$
where $\ell_Z (y) = \log \frac{f_1^Z (z)}{f_0^Z (z)}$, until 
$$ \sigma := \sigma^Z (A) = \min \{ n \geq 1: D_n \notin (0, A) \}. $$
If $D_{\sigma^Z} > A$, we stop and declare the change.

\item If $D_{\sigma^Z} < 0$, we start using observations $\{ Y_n \}$ and $\{ X_n \}$ from the lower-quality experiments and execute the truncated 2E-CUSUM algorithm for $Y$ and $X$ as follows:

    \begin{enumerate}
    \item We first scale the undershoot $D_\sigma$ by a factor $a_Z$, and use $a_Z D_{\sigma^Z}$ as the new zero level for the \textup{2E-CUSUM} algorithm for $Y$ and $X$.
    
    \item With $a_Z D_{\sigma^Z}$ as the new zero level, the stopping threshold is set at 0.
    
    \item We update the statistic $\{D_n\}$ starting at time $\sigma^Z (A)$ using the truncated \textup{2E-CUSUM} policy \\
    $\Psi_{\textup{2E}} (a_Z D_{\sigma^Z}, 0, a_Y, N_X, N_Y)$.
    
    \item When the policy $\Psi_{\textup{2E}} (a_Z D_{\sigma^Z}, 0, a_Y, N_X, N_Y)$ terminates, the statistic $D_n$ must cross 0 from below. At this time, we reset the statistic $D_n$ to $0$ and go to Step $1)$. 
    \end{enumerate}
\end{enumerate}    
\end{algorithm}

A typical evolution of the 3E-CUSUM algorithm is illustrated in Figure \ref{fig:3e-cusum}. For completeness, we describe the 3E-CUSUM algorithm in words. In this algorithm, the statistic $D_n$ starts at 0, and the decision-maker performs experiment $Z$. Thus, initially, the statistic $D_n$ evolves according to the CUSUM statistic from experiment $Z$ until either it crosses the threshold $A$ or it falls below $0$. If it crosses $A$, we stop and declare the change. If the statistic falls below $0$, we stop using experiment $Z$ and switch to experiment $Y$. 

\begin{figure}[!t]
\centering
\includegraphics[width = 0.5\linewidth]{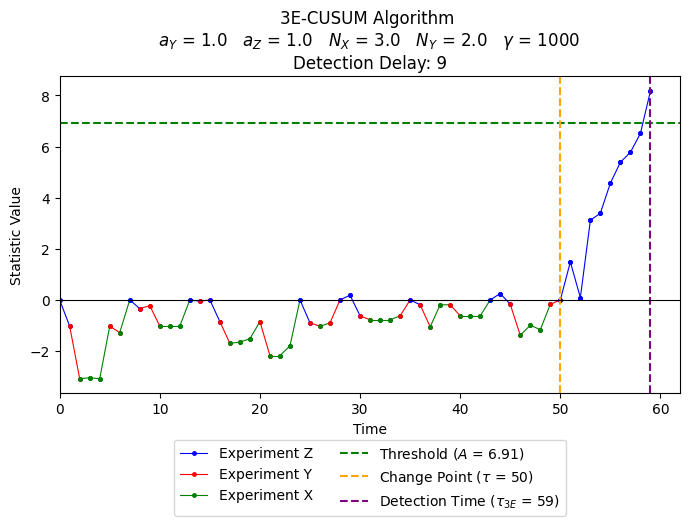}
\caption{An evolution of the statistic $D_n$ of the 3E-CUSUM algorithm for $f_0^X = f_0^Y = f_0^Z = \mathcal{N}(0, 1)$, $f_1^X = \mathcal{N}(0.5, 1)$, $f_1^Y = \mathcal{N}(0.75, 1)$, $f_1^Z = \mathcal{N}(1, 1)$, $\nu = 50$, and $\gamma = 1000$, with $a_Y = a_Z = 1.0$, $N_X = 3.0$, and $N_Y = 2.0$. The decision-maker performs experiment $Z$ whenever $D_n \geq 0$, experiment $Y$ when $D_n$ is below 0 and above its threshold, and experiment $X$ when $D_n$ is below the threshold for experiment $Y$.}
\label{fig:3e-cusum}
\end{figure}

The undershoot $D_{\sigma^Z}$ of the statistic $D_n$ when it goes below $0$ is scaled by a factor of $a_Z > 0$. This scaled value of the undershoot $ a_Z D_{\sigma^Z}$ becomes the new level zero for a CUSUM statistic update using experiment $Y$. The statistic $D_n$ is updated using the CUSUM statistic for $Y$, until either it goes above zero or goes below the new zero level $ a_Z D_{\sigma^Z}$. If the statistic goes above $0$, we reset it to $0$ and switch back to experiment $Z$. At this point, the process renews itself. 

If, while updating $D_n$ using the CUSUM update for $Y$, the statistic goes below $a_Z D_{\sigma^Z}$ (without ever becoming positive), then we switch to experiment $X$. The undershoot, say $U_Y$, at this time is scaled by $a_Y$, and $a_Z D_{\sigma^Z} + a_Y U_Y$ becomes the new level zero for the CUSUM statistic using experiment $X$. 

The update with experiment $X$ is continued, with statistic $D_n$ reflected at $a_Z D_{\sigma^Z} + a_Y U_Y$, until the statistic either goes above $ a_Z D_{\sigma^Z}$ or consumes $N_X$ observations. In both cases, we switch back to $Y$. 

Every time we start using the data from experiment $Y$, we start a counter with $N_Y$ and count down each time we use an observation of $Y$. While switching between $Y$ and $X$, if we consume $N_Y$ observations of $Y$, we switch back to $Z$. We note, however, that when using the $N_Y$th observation, we use the strategy, discussed before Algorithm~\ref{alg:3e-cusum}, employed by the truncated 2E-CUSUM algorithm.

When $N_X = 0$, the 3E-CUSUM algorithm simply becomes the 2E-CUSUM algorithm for experiments $Y$ and $Z$ (and experiment $X$ is never chosen). On the other hand, when $N_Y = 0$, 3E-CUSUM reduces to CUSUM for experiment $Z$.


\subsection{Achieving Different POR values}

We now show through simulations that we can design the 3E-CUSUM algorithm to achieve various choices of constraints on the fraction of time each experiment is performed. This fact is also established analytically in the next subsection using the Renewal Reward Theorem to analyze the POR metrics. 

We used Monte Carlo simulations to compute the PORs of 3E-CUSUM. Figure \ref{fig:por-ts-sim} and Table \ref{tab:por-ts-sim} show that various combinations of PORs can be achieved using 3E-CUSUM. Note that when $\beta_Y = 0$, then $\beta_X = 0$, i.e., if experiment $Y$ cannot be performed, then experiment $X$ cannot be chosen as well. When $\beta_X = 0$, then $\beta_Y \geq 0$, and experiment $Y$ can be performed even if experiment $X$ is not utilized.

\begin{figure}[!th]
\centering
\includegraphics[width = 0.5\linewidth]{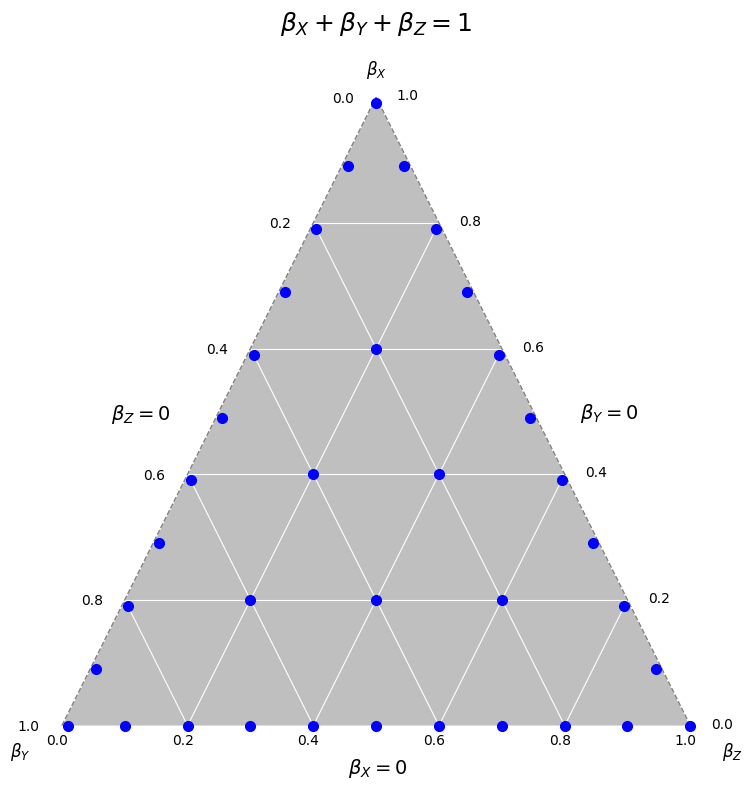}
\caption{Different combinations of target PORs for experiments $X$, $Y$, and $Z$. The blue dots represent combinations of target $\beta_X$, $\beta_Y$, and $\beta_Z$ for which we determined the parameters of the 3E-CUSUM algorithm, which approximately attains the target PORs (see Table \ref{tab:por-ts-sim}). The parameters of the 3E-CUSUM algorithm can be tuned to obtain any target PORs.}
\label{fig:por-ts-sim}
\end{figure}

\begin{table*}[!th]
\renewcommand{\arraystretch}{1.0}
\caption{Parameter values and POR results for the 3E-CUSUM algorithm given various target PORs for $f_0^X = f_0^Y = f_0^Z = \mathcal{N}(0, 1)$, $f_1^X = \mathcal{N}(0.5, 1)$, $f_1^Y = \mathcal{N}(0.75, 1)$, and $f_1^Z = \mathcal{N}(1, 1)$. The parameters of the 3E-CUSUM algorithm can be tuned to obtain any target PORs.}
\label{tab:por-ts-sim}
\centering
\begin{tabular}{ccc|cccc|ccc}
\multicolumn{3}{c}{Target} & \multicolumn{4}{c}{Parameters} & \multicolumn{3}{c}{Ouput} \\
\hline
$\beta_X$ & $\beta_Y$ & $\beta_Z$ & $a_Y$ & $a_Z$ & $N_X$ & $N_Y$ & $\textup{POR}_X$ & $\textup{POR}_Y$ & $\textup{POR}_Z$ \\
\hline
0.00 & 0.99 & 0.01 & 1.00 & 100.00 & 0.00 & 200.00 & 0.0000 & 0.9904 & 0.0096 \\
0.00 & 0.90 & 0.10 & 1.00 & 10.00 & 0.00 & 19.00 & 0.0000 & 0.9022 & 0.0978 \\
0.00 & 0.80 & 0.20 & 1.00 & 1.00 & 0.00 & 13.50 & 0.0000 & 0.8029 & 0.1971 \\
0.00 & 0.70 & 0.30 & 1.00 & 1.00 & 0.00 & 5.60 & 0.0000 & 0.6992 & 0.3008 \\
0.00 & 0.60 & 0.40 & 1.00 & 1.00 & 0.00 & 3.30 & 0.0000 & 0.6006 & 0.3994 \\
0.00 & 0.50 & 0.50 & 1.00 & 1.00 & 0.00 & 2.00 & 0.0000 & 0.5030 & 0.4970 \\
0.00 & 0.40 & 0.60 & 1.00 & 1.00 & 0.00 & 1.30 & 0.0000 & 0.4031 & 0.5969 \\
0.00 & 0.30 & 0.70 & 1.00 & 1.00 & 0.00 & 0.80 & 0.0000 & 0.2953 & 0.7047 \\
0.00 & 0.20 & 0.80 & 1.00 & 1.00 & 0.00 & 0.46 & 0.0000 & 0.1959 & 0.8041 \\
0.00 & 0.10 & 0.90 & 1.00 & 1.00 & 0.00 & 0.21 & 0.0000 & 0.0997 & 0.9003 \\
0.00 & 0.00 & 1.00 & 1.00 & 1.00 & 0.00 & 0.00 & 0.0000 & 0.0000 & 1.0000 \\
\hline
0.10 & 0.01 & 0.89 & 1.00 & 1.00 & 29.50 & 0.03 & 0.0996 & 0.0144 & 0.8860 \\
0.20 & 0.01 & 0.79 & 10.00 & 1.00 & 42.00 & 0.02 & 0.2003 & 0.0088 & 0.7910 \\
0.30 & 0.01 & 0.69 & 10.00 & 1.00 & 75.00 & 0.02 & 0.3007 & 0.0075 & 0.6918 \\
0.40 & 0.01 & 0.59 & 20.00 & 1.00 & 74.00 & 0.03 & 0.3976 & 0.0095 & 0.5929 \\
0.50 & 0.01 & 0.49 & 20.00 & 1.00 & 100.10 & 0.03 & 0.5002 & 0.0082 & 0.4916 \\
0.60 & 0.01 & 0.39 & 20.00 & 1.00 & 155.00 & 0.03 & 0.6034 & 0.0062 & 0.3903 \\
0.70 & 0.01 & 0.29 & 20.00 & 1.00 & 199.00 & 0.04 & 0.7050 & 0.0062 & 0.2889 \\
0.80 & 0.01 & 0.19 & 20.00 & 1.00 & 262.00 & 0.05 & 0.8009 & 0.0053 & 0.1939 \\
0.90 & 0.01 & 0.09 & 30.00 & 1.00 & 300.10 & 0.11 & 0.9035 & 0.0053 & 0.0912 \\
\hline
0.10 & 0.89 & 0.01 & 1.00 & 60.00 & 0.25 & 200.00 & 0.0966 & 0.8947 & 0.0087 \\
0.20 & 0.79 & 0.01 & 1.00 & 60.00 & 0.60 & 200.00 & 0.2044 & 0.7881 & 0.0076 \\
0.30 & 0.69 & 0.01 & 1.00 & 60.00 & 1.00 & 100.00 & 0.3001 & 0.6864 & 0.0134 \\
0.40 & 0.59 & 0.01 & 1.00 & 60.00 & 1.65 & 100.00 & 0.4011 & 0.5875 & 0.0115 \\
0.50 & 0.49 & 0.01 & 1.00 & 60.00 & 2.60 & 100.00 & 0.4993 & 0.4913 & 0.0094 \\
0.60 & 0.39 & 0.01 & 1.00 & 60.00 & 4.40 & 100.00 & 0.6010 & 0.3914 & 0.0076 \\
0.70 & 0.29 & 0.01 & 1.00 & 60.00 & 8.30 & 100.00 & 0.6998 & 0.2945 & 0.0057 \\
0.80 & 0.19 & 0.01 & 10.00 & 60.00 & 10.00 & 70.00 & 0.8023 & 0.1923 & 0.0054 \\
0.90 & 0.09 & 0.01 & 10.00 & 60.00 & 24.00 & 30.00 & 0.9028 & 0.0915 & 0.0057 \\
0.99 & 0.005 & 0.005 & 100.00 & 60.00 & 310.00 & 1.90 & 0.9893 & 0.0054 & 0.0053 \\
\hline
0.20 & 0.600 & 0.200 & 1.00 & 1.00 & 0.57 & 7.70 & 0.1974 & 0.5984 & 0.2042 \\
0.20 & 0.400 & 0.400 & 1.00 & 1.00 & 0.80 & 2.00 & 0.1978 & 0.4023 & 0.3999 \\
0.20 & 0.200 & 0.600 & 1.00 & 1.00 & 1.55 & 0.64 & 0.1967 & 0.2033 & 0.6000 \\
0.40 & 0.400 & 0.200 & 1.00 & 1.00 & 1.80 & 4.50 & 0.4007 & 0.3982 & 0.2010 \\
0.40 & 0.200 & 0.400 & 1.00 & 1.00 & 3.55 & 0.95 & 0.3973 & 0.2031 & 0.3997 \\
0.60 & 0.200 & 0.200 & 2.00 & 1.00 & 5.50 & 2.00 & 0.5986 & 0.2009 & 0.2005 \\
\hline
\end{tabular}
\end{table*}


\subsection{Design and Optimality of the 3E-CUSUM Algorithm}

We now show that we can set the parameters of the 3E-CUSUM algorithm to achieve any desired ARLFA and POR constraints. We also establish its asymptotic optimality. Recall that we defined in Algorithm \ref{alg:3e-cusum} $\sigma^Z(A) = \min \{ n \geq 1: D_n \notin (0, A) \}$, i.e., the first time $D_n$ falls below 0 or crosses $A$ from below. We let $\sigma^Z (\infty) = \sigma^Z$, i.e., the first time the statistic falls below 0. Note that in the 3E-CUSUM algorithm, the statistic $D_n$ is updated using log-likelihood ratios for experiment $Z$ while it is above $0$. When $D_{\sigma^Z} < 0$, we use $D_{\sigma^Z}$ to denote its undershoot. Also, we recall that $\tau_{\textup{\fontsize{6 pt}{10 pt}\selectfont C}}^Y (A)$ is the CUSUM stopping rule with observations from experiment $Y$. Figure \ref{fig:por-ts} illustrates these notations.

\begin{figure}[!t]
\centering
\includegraphics[scale=0.4]{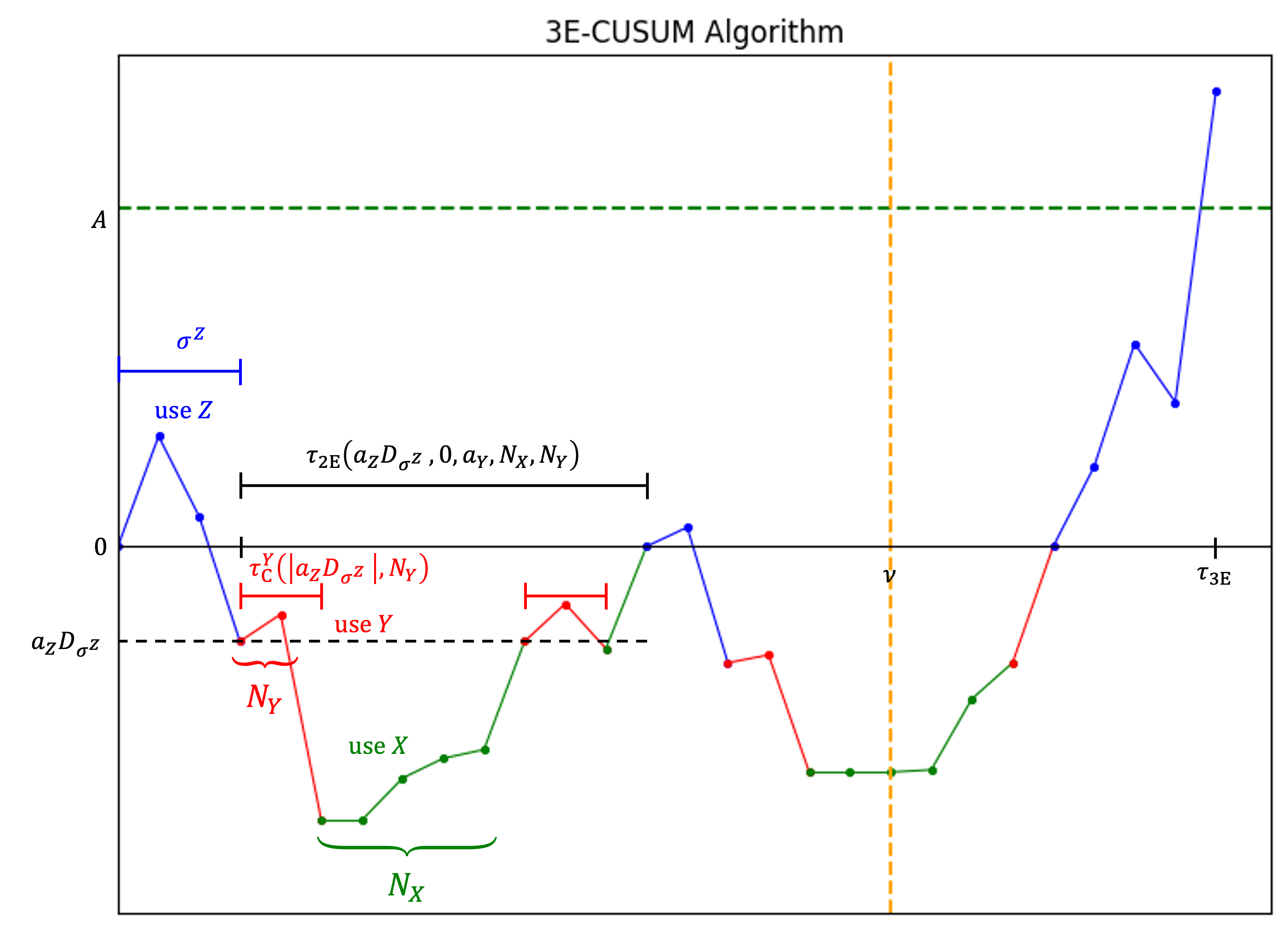}
\caption{The 3E-CUSUM Algorithm. Experiment $Z$ is performed first. When the statistic $D_n$ falls below 0, the undershoot is scaled by a factor of $a_Z$, and the truncated 2E-CUSUM algorithm for experiments $Y$ and $X$ is implemented. When the statistic crosses 0 from below, experiment $Z$ is performed again. Using pre-change distributions of $\{ X_n \}$, $\{ Y_n \}$, and $\{ Z_n \}$, $\sigma^Z$ is the length of time experiment $Z$ is performed, $\tau_{\textup{\fontsize{6 pt}{10 pt}\selectfont C}}^Y (\vert a_Z D_{\sigma^Z} \vert, N_Y)$ is the length of time experiment $Y$ is performed, and $\tau_{\textup{\fontsize{6 pt}{10 pt}\selectfont 2E}} (a_Z D_{\sigma^Z}, 0, N_X, N_Y)$ is the total length of time experiments $X$ and $Y$ are performed before change occurs at time $\nu$. The algorithm stops and declares change at $\tau_{\textup{\fontsize{6 pt}{10 pt}\selectfont 3E}}$ when the statistic crosses threshold $A$.}
\label{fig:por-ts}
\end{figure}

\medskip

\begin{theorem}
\label{thm:3e}

For the \textup{3E-CUSUM} Algorithm, let 
$$ 0 < D(f_1^X \mathrel{\Vert} f_0^X) \leq D(f_1^Y \mathrel{\Vert} f_0^Y) \leq D(f_1^Z \mathrel{\Vert} f_0^Z) < \infty. $$

\begin{enumerate}
\item For any fixed $a_Y > 0$, $a_Z > 0$, $N_X > 0$, and $N_Y > 0$, with $A = \log \gamma$,
$$ \textup{ARLFA} (\Psi_{\textup{3E}}) \geq \gamma. $$

\item For fixed values of $A$, $a_Y > 0$, $a_Z > 0$, $N_X > 0$, and $N_Y > 0$,
\begin{align*}
& \textup{POR}_Z (\tau_{\textup{\fontsize{6 pt}{10 pt}\selectfont 3E}})  = \frac{\mathsf{E}_\infty [\sigma^Z]}{\mathsf{E}_\infty [\sigma^Z] + \mathsf{E}_\infty [\tau_{\textup{\fontsize{6 pt}{10 pt}\selectfont 2E}}(a_Z D_{\sigma^Z}, 0, a_Y, N_X, N_Y)]} \\
& \textup{POR}_Y (\tau_{\textup{\fontsize{6 pt}{10 pt}\selectfont 3E}})  = \frac{\mathsf{E}_\infty [\min
\{ \tau_{\textup{\fontsize{6 pt}{10 pt}\selectfont C}}^Y (\vert a_Z D_{\sigma^Z}\vert), N_Y \}]}{\mathsf{E}_\infty [\sigma^Z] + \mathsf{E}_\infty [\tau_{\textup{\fontsize{6 pt}{10 pt}\selectfont 2E}}(a_Z D_{\sigma^Z}, 0, a_Y, N_X, N_Y)]}.
\end{align*}
The POR expressions are not a function of $A$. Thus, for any $A$, we can always choose values for $a_Y$, $a_Z$, $N_X$, and $N_Y$ to meet any given POR constraint of $\beta_Z$ and $\beta_Y$:
$$ \textup{POR}_Z (\Psi_{\textup{3E}} (A, a_Y, a_Z, N_X, N_Y)) \leq \beta_Z $$
and
$$ \textup{POR}_Y (\Psi_{\textup{3E}} (A, a_Y, a_Z, N_X, N_Y)) \leq \beta_Y. $$

\item For fixed values of $a_Y > 0$, $a_Z > 0$, $N_X > 0$, and $N_Y > 0$, and for each $A$,
$$ \textup{WADD} (\tau_{\textup{\fontsize{6 pt}{10 pt}\selectfont 3E}}) \leq \mathsf{E}_1 [\tau_{\textup{\fontsize{6 pt}{10 pt}\selectfont 3E}}] + N_Y + N_Y N_X. $$
Consequently, 
$$ \textup{WADD} (\Psi_{\textup{3E}}) \leq \textup{WADD} (\tau_{\textup{\fontsize{6 pt}{10 pt}\selectfont C}}^Z) + K_3, $$
where $\tau_{\textup{\fontsize{6 pt}{10 pt}\selectfont C}}^Z$ is the \textup{CUSUM} algorithm for the experiment $Z$ and $K_3$ is a constant that is not a function of threshold $A$. 
\end{enumerate}
\end{theorem}

\begin{proof}
\begin{enumerate}
\item The false alarm analysis for the 3E-CUSUM algorithm is similar to that of the 2E-CUSUM algorithm. Specifically, for any value of threshold $A$, we have
$$ \mathsf{E}_\infty [\tau_{\textup{\fontsize{6 pt}{10 pt}\selectfont 3E}}] \geq \mathsf{E}_\infty [\tau_{\textup{\fontsize{6 pt}{10 pt}\selectfont C}}^Z (A)] \geq e^A. $$

\item The expression for $\textup{POR}_Z$ follows again from the Renewal Reward Theorem because under the pre-change distribution, the expected number of observations of experiment $Z$ is $\mathsf{E}_\infty [\sigma^Z]$. Also, the average time spent below zero is given by the time for the truncated 2E-CUSUM algorithm to stop and bring the statistic above $0$: $\mathsf{E}_\infty [\tau_{\textup{\fontsize{6 pt}{10 pt}\selectfont 2E}}(a_Z D_{\sigma^Z}, 0, a_Y, N_X, N_Y)]$. This gives
\begin{align*}
\textup{POR}_Z (\tau_{\textup{\fontsize{6 pt}{10 pt}\selectfont 3E}})  = \frac{\mathsf{E}_\infty [\sigma^Z]}{\mathsf{E}_\infty [\sigma^Z] + \mathsf{E}_\infty [\tau_{\textup{\fontsize{6 pt}{10 pt}\selectfont 2E}}(a_Z D_{\sigma^Z}, 0, a_Y, N_X, N_Y)]}.
\end{align*}

When $N_Y = 0$, i.e., experiment $Y$ is never performed (consequently, experiment $X$ is never performed as well), $\textup{POR}_Z = 1$ and experiment $Z$ is utilized all the time. To decrease the fraction of time experiment $Z$ is chosen, we increase $N_Y$ and/or $N_X$ (and also $a_Y$ and $a_Z$), the number of times experiments $Y$ or $X$, respectively, can be done. Specifically, note that as $a_Z \to \infty$ and $N_Y \to \infty$, we have $\mathsf{E}_\infty [\tau_{\textup{\fontsize{6 pt}{10 pt}\selectfont 2E}}(a_Z D_{\sigma^Z}, 0, a_Y, N_X, N_Y)] \to \infty$ and, thus, $\textup{POR}_Z \to 0$.

\medskip

The expression for $\textup{POR}_Y$ also follows from the Renewal Reward Theorem because in the renewal cycles starting with statistic $D_n=0$ to the point where $D_n$ is again zero after being negative, the time spent using experiment $Y$ is equal to the time for a truncated (at $N_Y$) CUSUM statistic for experiment $Y$ to reach $0$ after starting and getting reflected at the negative level $a_Z D_{\sigma^Z}$. This time is statistically equal to $\min \{ \tau_{\textup{\fontsize{6 pt}{10 pt}\selectfont C}}^Y (\vert a_Z D_{\sigma^Z} \vert), N_Y \}$. Thus,
\begin{align*}
 \textup{POR}_Y (\tau_{\textup{\fontsize{6 pt}{10 pt}\selectfont 3E}}) = \frac{\mathsf{E}_\infty [\min \{ \tau_{\textup{\fontsize{6 pt}{10 pt}\selectfont C}}^Y (\vert a_Z D_{\sigma^Z} \vert), N_Y \}]}{\mathsf{E}_\infty [\sigma^Z] + \mathsf{E}_\infty [\tau_{\textup{\fontsize{6 pt}{10 pt}\selectfont 2E}}(a_Z D_{\sigma^Z}, 0, a_Y, N_X, N_Y)]}.
\end{align*}
To control $\textup{POR}_Y$, first note that as $N_X \to 0$,
\begin{align*}
 \textup{POR}_Y (\tau_{\textup{\fontsize{6 pt}{10 pt}\selectfont 3E}}) &= \frac{\mathsf{E}_\infty [\min \{ \tau_{\textup{\fontsize{6 pt}{10 pt}\selectfont C}}^Y (\vert a_Z D_{\sigma^Z} \vert), N_Y \}]}{\mathsf{E}_\infty [\sigma^Z] + \mathsf{E}_\infty [\tau_{\textup{\fontsize{6 pt}{10 pt}\selectfont 2E}}(a_Z D_{\sigma^Z}, 0, a_Y, N_X, N_Y)]} \\
& \to \frac{\mathsf{E}_\infty [\min \{ \tau_{\textup{\fontsize{6 pt}{10 pt}\selectfont C}}^Y (\vert a_Z D_{\sigma^Z} \vert), N_Y \}]}{\mathsf{E}_\infty [\sigma^Z] + \mathsf{E}_\infty [\min \{ \tau_{\textup{\fontsize{6 pt}{10 pt}\selectfont C}}^Y (\vert a_Z D_{\sigma^Z} \vert), N_Y \}]}.
\end{align*}
Thus, one way to increase $\textup{POR}_Y$ is to bring $N_X$ to $0$ and increase $a_Z$ and $N_Y$. On the other hand, to decrease $\textup{POR}_Y (\tau_{\textup{\fontsize{6 pt}{10 pt}\selectfont 3E}})$, we increase $N_X$ and $a_Y$. 

\medskip

Finally, we have
$$ \textup{POR}_X (\tau_{\textup{\fontsize{6 pt}{10 pt}\selectfont 3E}}) = 1 - \textup{POR}_Y (\tau_{\textup{\fontsize{6 pt}{10 pt}\selectfont 2E}}) - \textup{POR}_Z (\tau_{\textup{\fontsize{6 pt}{10 pt}\selectfont 2E}}). $$
Note that the expected number of observations of experiment $X$ in a cycle can be written as
\begin{align*}
 \mathsf{E}_\infty [\tau_{\textup{\fontsize{6 pt}{10 pt}\selectfont 2E}}(a_Z D_{\sigma^Z}, 0, a_Y, N_X, N_Y)]  - \mathsf{E}_\infty [\min \{ \tau_{\textup{\fontsize{6 pt}{10 pt}\selectfont C}}^Y (\vert a_Z D_{\sigma^Z} \vert), N_Y \}].
\end{align*}

\item For the delay analysis of the 3E-CUSUM algorithm, suppose the change occurs at time $\nu = 2$. Then the worst possible realization from experiment $Z$ at time $1$ would create a potentially infinite amount of undershoot $D_{\sigma^Z}$. Thus, the only way for the statistic $D_n$ to come back to $0$ would be to exhaust all the $N_Y$ samples from the experiment $Y$. Starting time 2, each time an observation from experiment $Y$ is chosen, there is a potential $N_X$ number of observations one can collect from experiment $X$. Thus, one can bound the number of samples before $D_n$ comes above zero by 
$$ N_Y + N_Y N_X. $$
Once the statistic $D_n$ comes above $0$, it is reset to $0$, and the delay from that point onward is 
$$ \mathsf{E}_1 [\tau_{\textup{\fontsize{6 pt}{10 pt}\selectfont 3E}}]. $$
A similar argument for other change points gives us 
$$ \textup{WADD} (\tau_{\textup{\fontsize{6 pt}{10 pt}\selectfont 3E}}) \leq \mathsf{E}_1 [\tau_{\textup{\fontsize{6 pt}{10 pt}\selectfont 3E}}] + N_Y + N_Y N_X. $$
Using an argument similar to that used in the case of the 2E-CUSUM algorithm, we have
\begin{align*}
\textup{WADD} (\tau_{\textup{\fontsize{6 pt}{10 pt}\selectfont 3E}})
& \leq \mathsf{E}_1 [\tau_{\textup{\fontsize{6 pt}{10 pt}\selectfont 3E}}] + N_Y + N_Y N_X \\
& = (\mathsf{E}_1 [\tau_{\textup{\fontsize{6 pt}{10 pt}\selectfont C}}^Z] + K) + N_Y + N_Y N_X \\
& = \textup{WADD} (\tau_{\textup{\fontsize{6 pt}{10 pt}\selectfont C}}^Z) + K + N_Y + N_Y N_X.
\end{align*}
Here $K$ is a constant that does not depend on the threshold $A$. 
\end{enumerate}
\end{proof}

The above theorem implies strong optimality properties of the 3E-CUSUM algorithm. 

\begin{corollary}
Under the conditions of the theorem, for a given $\gamma$, set $A = \log \gamma$. Then for any choice of $a_Y$, $a_Z$, $N_X$, and $N_Y$,
$$ \textup{ARLFA} (\Psi_{\textup{3E}}) \geq \gamma. $$
For a given $\beta_Y$ and $\beta_Z$, we can find parameters $a_Y$, $a_Z$, $N_X$, and $N_Y$ such that for all $A$ (specifically, $A = \log \gamma$),
$$ \textup{POR}_Z (\Psi_{\textup{3E}}) \leq \beta_Z, $$
and
$$ \textup{POR}_Y (\Psi_{\textup{3E}}) \leq \beta_Y. $$
Moreover, for such a choice of these parameters, 
\begin{align*}
\textup{WADD} (\tau_{\textup{\fontsize{6 pt}{10 pt}\selectfont 3E}}) \sim \textup{WADD} (\tau_{\textup{\fontsize{6 pt}{10 pt}\selectfont C}}^Z) \sim \frac{\log \gamma}{D(f_1^Z \mathrel{\Vert} f_0^Z)} \;\; \text{as} \;\; \gamma \to \infty. 
\end{align*}
\end{corollary}

\begin{proof}
The proof follows from Theorem \ref{thm:3e} and the property of the CUSUM algorithm. 
\end{proof}


\subsection{Truncated \textup{3E-CUSUM} Algorithm}
\label{sec:trunc-3e}

In this section, we define a truncated version of the 3E-CUSUM algorithm in a manner similar to how we defined the truncated 2E-CUSUM in Section \ref{sec:trunc-2e}. As discussed there, we need this truncated modification to define QCD with four experiments. 

Notice that in the 3E-CUSUM algorithm, there is no limit for the observations on experiment $Z$. However, there is a limit of $N_Y$ for observations on experiment $Y$, each time it is used before reverting to experiment $Z$. There is also a limit of $N_X$ for observations on experiment $X$, each time it is used before reverting to experiment $Y$. In the truncated 3E-CUSUM algorithm, we simply impose a limit of $N_Z$ for observations on experiment $Z$. While performing experiment $Z$, if the statistic crosses the detection threshold or $N_Z$ observations are reached, we stop and declare a change. For completeness, we provide the full description of the truncated 3E-CUSUM algorithm below. 

\begin{algorithm}[Truncated 3E-CUSUM: $\Psi_{\textup{3E}} (A, a_Y, a_Z, N_X, N_Y, N_Z)$]
\label{alg:trunc-3e}

Fix $A > 0$, $a_Y > 0$, $a_Z > 0$, $N_X \geq 0$, $N_Y \geq 0$, and $N_Z \geq 0$. The truncated \textup{3E-CUSUM} statistic $\{ D_n \}$ is defined as follows:

\begin{enumerate}
\item We start with $D_0 = 0$ and update $\{ D_n \}$ using the observations $\{ Z_n \}$ from the highest quality experiment using
$$ D_{n + 1}  = D_n + \ell_Z (Z_{n + 1}), $$
where $\ell_Z (z) = \log \frac{f_1^Z (z)}{f_0^Z (z)}$, until either time
$$ \sigma := \sigma^Z(A) = \min \{ n \geq 1: D_n \notin (0, A) \} $$
or $N_Z$ number of observations of experiment $Z$ is used. If $D_{\sigma^Z} > A$ or $N_Z$ has been reached, we stop and declare the change.

\item If $D_{\sigma^Z} < 0$, we start using observations $\{ Y_n \}$ and $\{ X_n \}$ from the lower-quality experiments and execute the truncated \textup{2E-CUSUM} algorithm for $Y$ and $X$ as follows:
    
    \begin{enumerate}
    \item We first scale the undershoot $D_{\sigma^Z}$ by a factor $a_Z$, and use $a_Z D_{\sigma^Z}$ as the new zero level for the truncated \textup{2E-CUSUM} algorithm for $Y$ and $X$.
    
    \item With $a_Z D_{\sigma^Z}$ as the new zero level, the stopping threshold is set at $0$.
    
    \item We update the statistic $\{ D_n \}$ starting at time $\sigma^Z(A)$ using the truncated \textup{2E-CUSUM} policy \\ $\Psi_{\textup{2E}} (a_Z D_{\sigma^Z}, 0, a_Y, N_X, N_Y)$.
    
    \item When the policy $\Psi_{\textup{2E}} (a_Z D_{\sigma^Z}, 0, a_Y, N_X, N_Y)$ terminates, the statistic $D_n$ must cross 0 from below. At this time, we reset the statistic $D_n$ to $0$ and go to Step $1)$.
    \end{enumerate}
\end{enumerate}    
\end{algorithm}


\section{QCD with Multiple-Experiment Design}
\label{sec:multiple-expt}

At this point, we are now finally ready to solve the multiple-experiment QCD problem we formulated in Section~\ref{sec:prob}. The analyses we did for the two- and three-experiment cases will help us generalize our proposed algorithms to the $m$-experiment case.

Observe that the 3E-CUSUM algorithm starts by using experiment $Z$, the highest-quality experiment, while the statistic is above 0. Once it falls below 0, we scale the undershoot and use this scaled value as the starting point of the truncated 2E-CUSUM algorithm for experiments $Y$ and $X$. Once the statistic crosses 0 from below, we use experiment $Z$ again.

In a similar vein, we can propose the 4E-CUSUM algorithm as follows: while the statistic is above 0, the highest-quality experiment is performed. When the statistic falls below 0, the scaled undershoot is used as the starting point of the truncated 3E-CUSUM algorithm. Once the statistic goes above 0, the highest-quality experiment is used again.

Following this recursive description, we propose the $m$E-CUSUM algorithm ($\Psi_{\textup{$m$E}} (A, a_2, a_3, \dots, a_m, N_1, N_2, \dots, N_{m - 1})$). This algorithm starts by using experiment $m$, the highest-quality experiment, while the statistic is above 0. Once it falls below 0, we scale the undershoot and use this value as the starting point of the truncated $(m - 1)$E-CUSUM algorithm for the remaining $m - 1$ experiments. Once the statistic crosses 0 from below, we perform experiment $m$ again.

\begin{algorithm}[$m$E-CUSUM: $\Psi_{\textup{$m$E}} (A, a_2, \dots, a_m, N_1, \dots, N_{m - 1})$]
\label{alg:me-cusum}

Fix $A > 0$, $a_i > 0$ ($i = 2, \dots, m$), and $N_j \geq 0$ ($j = 1, \dots, m - 1$). The \textup{$m$E-CUSUM statistic} $\{ D_n \}$ is defined as follows:

    \begin{enumerate}
    \item We start with $D_0 = 0$ and update $\{ D_n \}$ using the observations $\{ X_n^m \}$ from experiment $m$, the highest quality experiment, using
    $$ D_{n + 1}  = D_n + \ell_m (X_{n + 1}^m), $$
    where $\ell_m (x^m) = \log \frac{f_1^m (x^m)}{f_0^m (x^m)}$, until 
    $$ \sigma := \sigma^m (A) = \min \{ n \geq 1: D_n \notin (0, A) \}. $$
    If $D_{\sigma^m} > A$, we stop and declare the change.
    
    \item If $D_{\sigma^m} < 0$, we start using observations $\{ X_n^{m - 1} \}$, $\{ X_n^{m - 2} \}$, $\dots$, $\{ X_n^1 \}$ from the lower-quality experiments and execute the truncated \textup{$(m - 1)$E-CUSUM} algorithm for experiment $m - 1$, experiment $m - 2$, up to experiment 1 as follows:
    
    \begin{enumerate}
    \item We first scale the undershoot $D_{\sigma^m}$ by a factor of $a_m$, and use $a_m D_{\sigma^m}$ as the new zero level for the \textup{$(m - 1)$E-CUSUM} algorithm for experiment $m - 1$, experiment $m - 2$, up to experiment 1.
    
    \item With $a_m D_{\sigma^m}$ as the new zero level, the stopping threshold is set at $0$.
    
    \item We update the statistic $\{ D_n \}$ starting at time $\sigma^m (A)$ using the truncated $(m - 1)$\textup{E-CUSUM} policy \\ $\Psi_{\textup{$(m - 1)$E}} (a_m D_{\sigma^m}, 0, a_2, \dots, a_m, N_1, \dots, N_{m - 1})$.
    
    \item When the policy $\Psi_{\textup{$(m - 1)$E}} (a_m D_{\sigma^m}, 0, a_2, \dots, a_m, N_1, \dots, N_{m - 1})$ terminates, the statistic $D_n$ must cross 0 from below. At this time, we reset the statistic $D_n$ to $0$ and go to Step $1)$.
    \end{enumerate}
\end{enumerate}    
\end{algorithm}

\medskip

In the theorem below, we state the main result on the optimality of the $m$E-CUSUM algorithm. We skip the proof as the proof is similar to that of the 3E-CUSUM algorithm, but is notationally complex. 

\begin{theorem}
\label{thm:me}

For the \textup{$m$E-CUSUM} Algorithm, let $0 < D(f_1^j \mathrel{\Vert} f_0^j) < D(f_0^i \mathrel{\Vert} f_1^i) < \infty$ for $i > j$ where $i, j = 1, 2, \dots, m$.

\begin{enumerate}
\item For any fixed $a_i > 0$ ($i = 2, 3, \dots, m$) and $N_j > 0$ ($j = 1, 2, \dots, m - 1$), with $A = \log \gamma$, $\textup{ARLFA} (\Psi_{\textup{$m$E}}) \geq \gamma$.

\item For any $A$, we can choose the parameters such that
$ \textup{POR}_i (\Psi_{\textup{$m$E}} (A, a_2, a_3, \dots, a_m, N_1, N_2, \dots, N_{m - 1})) \leq \beta_i $
for $i = 2, 3, \dots, m$.

\item For fixed values of $a_i > 0$ ($i = 2, 3, \dots, m$) and $N_j > 0$ ($j = 1, 2, \dots, m - 1$), and for each $A$,
$$ \textup{WADD} (\tau_{\textup{\fontsize{6 pt}{10 pt}\selectfont $m$E}}) \leq \mathsf{E}_1 [\tau_{\textup{\fontsize{6 pt}{10 pt}\selectfont $m$E}}] + \sum_{r = 1}^{m - 1} \prod_{k = m - r}^{m - 1} N_k. $$
Consequently,
$$ \textup{WADD} (\Psi_{\textup{$m$E}}) \leq \textup{WADD} (\tau_{\textup{\fontsize{6 pt}{10 pt}\selectfont C}}^m) + K_m, $$
where $\tau_{\textup{\fontsize{6 pt}{10 pt}\selectfont C}}^m$ is the \textup{CUSUM} algorithm for the experiment $m$ and $K_m$ is a constant that is not a function of threshold $A$. Hence, the \textup{$m$E-CUSUM} Algorithm is asymptotically optimal.
\end{enumerate}

\end{theorem}

\clearpage

\section{QCD with Experiment Design and Data Efficiency}
\label{sec:data-efficient}

In this section, we address the issue of data efficiency in the two-experiment QCD problem before generalizing to the case of $m$ experiments. In data efficiency, the decision-maker has the additional option not to perform any experiment.

\subsection{Data-Efficient Two-Experiment Design}
\label{sec:de-two-expt}

In the data-efficient two-experiment QCD problem, the decision-maker has two experiments it can perform: the higher quality experiment $Y$ and the lower quality experiment $X$. They also have the option not to perform any experiment. The problem formulation is similar to the original two-experiment QCD problem: 
\begin{align}
\begin{split}
\min_{\Psi} \quad & \textup{WADD} (\Psi), \\
\text{s.t. } \quad & \textup{ARLFA} (\Psi) \geq \gamma, \\
& \textup{POR}_Y (\Psi) \leq \beta_Y, \\
& \textup{POR}_X (\Psi) \leq \beta_X,
\end{split}
\end{align}
where $\gamma \geq 0$, $0 < \beta_Y < 1$, and $0 \leq \beta_X < 1$ are given constants. If $\beta_Y + \beta_X=1$, then this problem reduces to the problem in equation~\eqref{eq:2e-qcd}. Thus, in this problem, we assume that $\beta_Y + \beta_X < 1$. The gap $\beta_0 = 1 - \beta_Y - \beta_X$ is the constraint on the fraction of time neither experiment $Y$ nor $X$ is performed. 

Our solution to this problem is similar to the three-experiment design, except we replace the experiment with the lowest information quality with the case where no experiment is performed. Instead of utilizing a truncated test at the level for a third experiment, a truncated test is done at this level with no experiment, and the statistic is increased deterministically using a design parameter $\mu$. 
We propose the Data-Efficient 2E-CUSUM (DE2E-CUSUM) algorithm.

\begin{algorithm}[DE2E-CUSUM: $\Psi_{\textup{DE2E}} (A, a_X, a_Y, N_0, N_X, \mu)$]
\label{alg:de2e-cusum}

Fix $A > 0$, $a_X > 0$, $a_Y > 0$, $N_0 \geq 0$, $N_X \geq 0$, and $\mu > 0$. The \textup{DE2E-CUSUM} statistic $\{ D_n \}$ is defined as follows:

\begin{enumerate}
\item We start with $D_0 = 0$ and update $\{ D_n \}$ using the observations $\{ Y_n \}$ from the highest quality experiment using
$$ D_{n + 1}  = D_n + \ell_Y (Y_{n + 1}), $$
where $\ell_Y (y) = \log \frac{f_1^Y (y)}{f_0^Y (y)}$, until 
$$ \sigma := \sigma^Y (A) = \min \{ n \geq 1: D_n \notin (0, A) \}. $$
If $D_{\sigma^Y} > A$, we stop and declare the change.

\item If $D_{\sigma^Y} < 0$, we start using observations $\{ X_n \}$ from the lower-quality experiment and execute a modified truncated 2E-CUSUM algorithm for experiment $X$ and experiment 0 (where no experiment is performed) as follows:

    \begin{enumerate}
    \item We first scale the undershoot $U_Y:= D_{\sigma^Y}$ by a factor $a_Y$, and use $a_Y U_Y$ as the new zero level for the modified truncated \textup{2E-CUSUM} algorithm for experiment $X$ and experiment 0.
    
    \item With $a_Y U_Y$ as the new zero level, the stopping threshold is set at $0$.
    
    \item We update the statistic $\{ D_n \}$ starting at time $\sigma$ using
    $$ D_{n + 1} = D_n + \ell_X (X_{n + 1}) $$
    where $\ell_X (X_{n + 1}) = \log \frac{f_1^X (x)}{f_0^X (x)}$ until either the statistic falls below $a_Y U_Y$ at time
    $$ \sigma^X (a_Y U_Y) = \min \{ n \geq \sigma^Y + 1: D_n \not \in (a_Y U_Y, 0) \} $$
    or $N_X$ number of observations of $X$ are used. If $D_{\sigma^X (a_Y U_Y)} > 0$, we reset the statistic $D_n$ to $0$ and go to Step $1)$.
    
    \item If $D_{\sigma^X (a_Y U_Y)} < a_Y U_Y$, we do not perform any experiment but still scale the undershoot $U_X = D_{\sigma^X (a_Y U_Y)}$ by a factor of $a_X$ and use $a_X U_X$ as the starting level for deterministically updating $D_n$ as
    $$ D_{n + 1} = D_n + \mu. $$
    This continues until either $D_n$ goes above $a_Y U_Y$ or $N_0$ number of deterministic updates are used. At this time, we reset $D_n$ to $a_Y U_Y$ and continue updating $D_n$ as in step $\textup{c)}$.
    \end{enumerate}

\item Stopping Rule: Stop at the stopping time
$$ \tau_{\textup{\fontsize{6 pt}{10 pt}\selectfont DE2E}} = \inf \{ n \geq 1: D_n > A \}. $$
\end{enumerate}    
\end{algorithm}

It is clear from the description of the DE2E-CUSUM algorithm that it is a special/degenerate case of the 3E-CUSUM Algorithm discussed in Section~\ref{sec:3e-cusum}. To clarify this point, suppose a decision-maker is executing the 3E-CUSUM algorithm and has to choose between three experiments $Z$, $Y$, and $X$, with $Z$ being of the highest quality. 
The decision-maker uses the experiment $Z$ and $Y$ in the usual way. 
However, when the decision-maker has the opportunity to utilize the lowest quality experiment $X$ (in the 3E-CUSUM algorithm), it does not use it and simply chooses not to perform any experiment at all. The statistic is then updated deterministically by adding to it a design parameter $\mu$ every time step when no experiment is performed. 
Thus, the DE2E-CUSUM algorithm is identical in structure to the 3E-CUSUM algorithm, with the experiment $Z$ and $Y$ in the 3E-CUSUM algorithm replaced by $Y$ and $X$ in the DE2E-CUSUM algorithm. 
A typical evolution of DE2E-CUSUM is illustrated in Figure \ref{fig:de2e-cusum}.

\begin{figure}[!t]
\centering
\includegraphics[width = 0.5\linewidth]{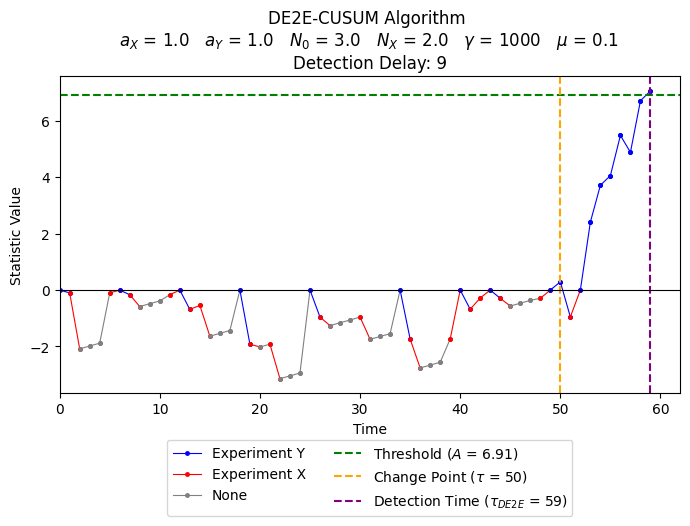}
\caption{An evolution of the statistic $D_n$ of the DE2E-CUSUM algorithm for $f_0^X = f_0^Y = \mathcal{N}(0, 1)$, $f_1^X = \mathcal{N}(0.75, 1)$, $f_1^Y = \mathcal{N}(1, 1)$, $\nu = 50$, and $\gamma = 1000$, with $a_X = a_Y = 1.0$, $N_X = 2.0$, $N_0 = 3.0$, and $\mu = 0.1$. The decision-maker performs experiment $Y$ whenever $D_n \geq 0$, and may perform experiment $X$ or not perform any experiment at all when $D_n < 0$.}
\label{fig:de2e-cusum}
\end{figure}

We now state a theorem on the design and optimality of the DE2E-CUSUM algorithm, which states that we can set the parameters of the algorithm to achieve any desired ARLFA and POR constraints. Also, the algorithm is asymptotically optimal since its delay is within a constant of the delay of the CUSUM algorithm using the highest quality experiment $Y$. The proof is skipped since it is similar to the proof of Theorem \ref{thm:3e}.

\begin{theorem}
\label{thm:de2e}

For the \textup{DE2E-CUSUM} algorithm, let $0 < D(f_0^X \mathrel{\Vert} f_1^X) \leq D(f_0^Y \mathrel{\Vert} f_1^Y) < \infty$.

\begin{enumerate}
\item For any fixed $a_X > 0$, $a_Y > 0$, $N_0 > 0$, $N_X > 0$, and $\mu > 0$, with $A = \log \gamma$, $\textup{ARLFA} (\Psi_{\textup{DE2E}}) \geq \gamma$.

\item For any $A$, we can choose the parameters such that
\begin{align*}
& \textup{POR}_Y (\Psi_{\textup{DE2E}} (A, a_X, a_Y, N_0, N_X, \mu)) \leq \beta_Y \text{ and } \\
& \textup{POR}_X (\Psi_{\textup{DE2E}} (A, a_X, a_Y, N_0, N_X, \mu)) \leq \beta_X.
\end{align*}
If $\beta_X + \beta_Y < 1$, then the data-efficiency constraint is also satisfied.

\item For fixed values of $a_X > 0$, $a_Y > 0$, $N_0 > 0$, $N_X > 0$, and $\mu > 0$; and for each $A$,
$$ \textup{WADD} (\tau_{\textup{\fontsize{6 pt}{10 pt}\selectfont DE2E}}) \leq \mathsf{E}_1 [\tau_{\textup{\fontsize{6 pt}{10 pt}\selectfont DE2E}}] + N_X + N_X N_0. $$
Consequently, 
$$ \textup{WADD} (\Psi_{\textup{DE2E}}) \leq \textup{WADD} (\tau_{\textup{\fontsize{6 pt}{10 pt}\selectfont C}}^Y) + K_2, $$
where $\tau_{\textup{\fontsize{6 pt}{10 pt}\selectfont C}}^Y$ is the \textup{CUSUM} algorithm for the experiment $Y$ and $K_2$ is a constant that is not a function of threshold $A$. Hence, the \textup{DE2E-CUSUM} Algorithm is asymptotically optimal.
\end{enumerate}

\end{theorem}




\subsection{Truncated \textup{DE2E-CUSUM} Algorithm}
Recall from Section~\ref{sec:3e-cusum} and Section~\ref{sec:multiple-expt} that we defined, in an iterative manner, the 3E-CUSUM algorithm using the truncated 2E-CUSUM algorithm, defined the 4E-CUSUM algorithm using the truncated 3E-CUSUM algorithm, and finally, the $m$E-CUSUM algorithm using the truncated 
$m-1$E-CUSUM algorithm. We use a similar approach to solve the problem of QCD with multiple experiments and data efficiency. The key to this development is the definition of the truncated DE2E-CUSUM algorithm. 

In the truncated DE2E-CUSUM algorithm, we impose a limit on the allowed number of consecutive observations of the highest quality experiment $Y$. We provide the full description of the algorithm below for completeness. 

\clearpage

\begin{algorithm}[Truncated DE2E-CUSUM: $\Psi_{\textup{DE2E}} (A, a_X, a_Y, N_0, N_X, N_Y, \mu)$]
\label{alg:trunc-de2e}

Fix $A > 0$, $a_X > 0$, $a_Y > 0$, $N_0 \geq 0$, $N_X \geq 0$, $N_Y \geq 0$, and $\mu > 0$. The \textup{DE2E-CUSUM} statistic $\{ D_n \}$ is defined as follows:

\begin{enumerate}
\item We start with $D_0 = 0$ and update $\{ D_n \}$ using the observations $\{ Y_n \}$ from the highest quality experiment using
$$ D_{n + 1}  = D_n + \ell_Y (Y_{n + 1}), $$
where $\ell_Y (y) = \log \frac{f_1^Y (y)}{f_0^Y (y)}$, until either time
$$ \sigma := \sigma^Y (A) = \min \{ n \geq 1: D_n \notin (0, A) \} $$
or $N_Y$ number of observations of experiment $Y$ is used. If $D_{\sigma^Y} > A$ or $N_Y$ has been reached, we stop and declare the change.

\item If $D_{\sigma^Y} < 0$, we start using observations $\{ X_n \}$ from the lower-quality experiment and execute a modified truncated \textup{2E-CUSUM} algorithm for experiment $X$ and experiment 0 (where no experiment is performed) as follows:

    \begin{enumerate}
    \item We first scale the undershoot $U_Y:= D_{\sigma^Y}$ by a factor $a_Y$, and use $a_Y U_Y$ as the new zero level for the modified truncated 2E-CUSUM algorithm for experiment $X$ and experiment 0.
    
    \item With $a_Y U_Y$ as the new zero level, the stopping threshold is set at $0$.
    
    \item We update the statistic $\{ D_n \}$ starting at time $\sigma$ using
    $$ D_{n + 1} = D_n + \ell_X (X_{n + 1}) $$
    where $\ell_X (X_{n + 1}) = \log \frac{f_1^X (x)}{f_0^X (x)}$ until either the statistic falls below $a_Y U_Y$ at time
    $$ \sigma^X (a_Y U_Y) = \min \{ n \geq \sigma^Y + 1: D_n \notin (a_Y U_Y, 0) \} $$
    or $N_X$ number of observations of $X$ are used. If $D_{\sigma^X (a_Y U_Y)} > 0$, we reset the statistic $D_n$ to $0$ and go to Step $1)$.
    
    \item If $D_{\sigma^X (a_Y U_Y)} < a_Y U_Y$, we do not perform any experiment, but still scale the undershoot $U_X = D_{\sigma^X (a_Y U_Y)}$ by a factor of $a_X$ and use $a_X U_X$ as the starting level for deterministically updating $D_n$ as
    $$ D_{n + 1} = D_n + \mu. $$
    This continues until either $D_n$ goes above $a_Y U_Y$ or $N_0$ number of deterministic updates are used. At this time, we reset $D_n$ to $a_Y U_Y$ and continue updating $D_n$ as in step $\textup{c)}$.
    \end{enumerate}

\item Stopping Rule: Stop at the stopping time
$$ \tau_{\textup{\fontsize{6 pt}{10 pt}\selectfont DE2E}}(A, a_X, a_Y, N_0, N_X, N_Y, \mu) = \inf \{ n \geq 1: D_n > A \}. $$
\end{enumerate}    
\end{algorithm}
It is clear from the description that 
$$ \tau_{\textup{\fontsize{6 pt}{10 pt}\selectfont DE2E}} = \tau_{\textup{\fontsize{6 pt}{10 pt}\selectfont DE2E}}(A, a_X, a_Y, N_0, N_X, \infty, \mu). $$


\subsection{Data-Efficient Multiple-Experiment Design}

In the data-efficient $m$-experiment QCD problem, the decision-maker has $m$ experiments to perform. It also has the option not to perform any experiment. The problem formulation is as follows:
\begin{align}
\begin{split}
\min_{\Psi} \quad & \textup{WADD} (\Psi) \\
\text{s.t. } \quad & \textup{ARLFA} (\Psi) \geq \gamma \\
& \textup{POR}_i (\Psi) \leq \beta_i \text{ for } i = 1, \dots, m
\end{split}
\end{align}
where $\gamma \geq 0$, $0 < \beta_m < 1$, and $0 \leq \beta_i < 1$ for $i = 2, \dots, m - 1$ are given constants. If $\ds \sum_{i = 1}^m \beta_i = 1$, then this problem reduces to the problem in equation~\eqref{eq:mE-qcd}. Thus, in this problem, we assume that
$$ \sum_{i = 1}^m \beta_i < 1. $$
The gap
$$ \beta_0 = 1 - \sum_{i = 1}^m \beta_i $$
is the constraint on the fraction of time no experiments are performed. We now discuss our algorithm that solves this problem and which is defined using the DE2E-CUSUM algorithm in an iterative manner.



We first propose the Data-Efficient 3E-CUSUM (DE3E-CUSUM) algorithm ($\Psi_{\textup{DE3E}} (A, a_X, a_Y, a_Z, N_0, N_X, N_Y, \mu$). This algorithm starts by using experiment $Z$, the highest-quality experiment, while the statistic is above 0. Once it falls below 0, we scale the undershoot and use this scaled value as the starting point of the truncated DE2E-CUSUM Algorithm for experiments $Y$ and $X$ until the statistic crosses 0 from below. At this point, we start using experiment $Z$ again. A typical evolution of DE3E-CUSUM is illustrated in Figure \ref{fig:de3e-cusum}. A truncated DE3E-CUSUM algorithm is one where the use of the highest quality experiment $Z$ is truncated to a fixed number, say $N_Z$, of samples before a change is declared. 

\begin{figure}[!t]
\centering
\includegraphics[scale=0.6]{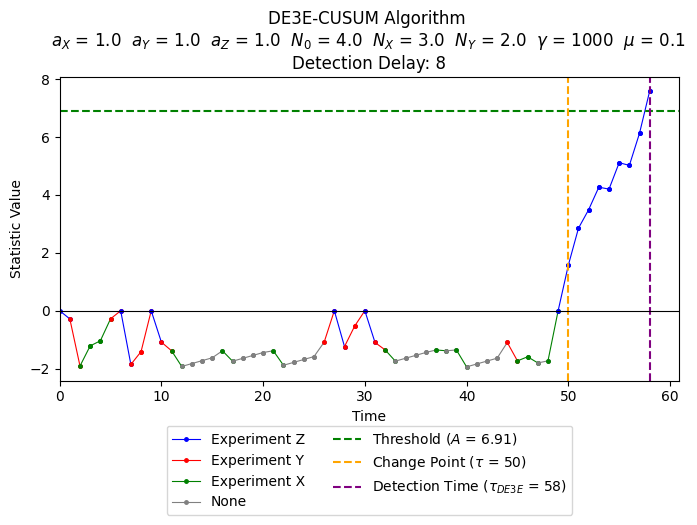}
\caption{An evolution of the statistic $D_n$ of the DE3E-CUSUM algorithm for $f_0^X = f_0^Y = f_0^Z = \mathcal{N}(0, 1)$, $f_1^X = \mathcal{N}(0.5, 1)$, $f_1^Y = \mathcal{N}(0.75, 1)$, $f_1^Z = \mathcal{N}(1, 1)$, $\nu = 50$, and $\gamma = 1000$, with $a_X = a_Y = a_Z = 1.0$, $N_X = 3.0$, $N_Y = 2.0$, $N_0 = 4.0$, and $\mu = 0.1$. The decision-maker performs experiment $Z$ whenever $D_n \geq 0$. When $D_n < 0$, the observer performs experiments $Y$ and $X$ within their thresholds, or does not perform any experiment after experiment $X$.}
\label{fig:de3e-cusum}
\end{figure}

Similar to the DE3E-CUSUM algorithm, the DE4E-CUSUM algorithm is defined as follows. In this algorithm, while the statistic is above 0, the highest-quality experiment is performed. Once the statistic falls below 0, the undershoot is scaled, and the truncated DE3E-CUSUM algorithm is used. 
Once the statistic crosses 0 from below, we perform the highest-quality experiment again.

Finally, the Data-Efficient $m$E-CUSUM (DE$m$E-CUSUM) algorithm ($\Psi_{\textup{DE$m$E}} (A, a_1, a_2, \dots, a_m, N_0, N_1, \dots, N_{m - 1}, \mu)$) is proposed as follows. The algorithm starts with the highest-quality experiment, which is used while the statistic is above 0. When it falls below 0, {\color{black} the undershoot is scaled and this scaled value is used as the starting point of} the truncated DE$(m - 1)$E-CUSUM Algorithm for experiments 1 to $m - 1$. The highest-quality experiment is performed again once the statistic crosses 0 from below.

\begin{algorithm}[DE$m$E-CUSUM: $\Psi_{\textup{DE$m$E}} (A, a_1, a_2, \dots, a_m, N_0, N_1, \dots, N_{m - 1}, \mu)$]
\label{alg:deme-cusum}

Fix $A > 0$, $a_i > 0$ ($i = 1, 2, \dots, m$), $N_j \geq 0$ ($j = 0, 1, \dots, m - 1$), and $\mu > 0$. The \textup{DE$m$E-CUSUM} statistic $\{ D_n \}$ is defined as follows:

\begin{enumerate}
\item We start with $D_0 = 0$ and update $\{ D_n \}$ using the observations $\{ X_n^m \}$ from experiment $m$, the highest quality experiment, using
$$ D_{n + 1}  = D_n + \ell_m (X_{n + 1}^m), $$
where $\ell_m (x^m) = \log \frac{f_1^m (x^m)}{f_0^m (x^m)}$, until 
$$ \sigma := \sigma^m (A) = \min \{ n \geq 1: D_n \notin (0, A) \}. $$
If $D_{\sigma^m} > A$, we stop and declare the change.

\item If $D_{\sigma^m} < 0$, we start using observations $\{ X_n^{m - 1} \}$, $\{ X_n^{m - 2} \}$, \dots, $\{ X_n^1 \}$ from the lower-quality experiments and execute the truncated \textup{DE$(m - 1)$E-CUSUM} algorithm for experiment $m - 1$, experiment $m - 2$, up to experiment 1 as follows:

    \begin{enumerate}
    \item We first scale the undershoot $U_m:= D_{\sigma^m}$ by a factor of $a_m$, and use $a_m U_m$ as the new zero level for the truncated \textup{DE$(m - 1)$E-CUSUM} algorithm for experiment $m - 1$, experiment $m - 2$, up to experiment 1.
    
    \item With $a_m U_m$ as the new zero level, the stopping threshold is set at $0$.
    
    \item We update the statistic $\{ D_n \}$ starting at time $\sigma^m$ using the truncated \textup{DE$(m - 1)$E-CUSUM} policy \\ $\Psi_{\textup{DE$(m - 1)$E}} (a_m U_m, 0, a_1, \dots, a_{m - 1}, N_0, \dots, N_{m - 1}, \mu)$.
    
    \item When the policy $\Psi_{\textup{DE$(m - 1)$E}} (a_m U_m, 0, a_1, \dots, a_{m - 1}, N_0, \dots, N_{m - 1}, \mu)$ terminates, the statistic $D_n$ must cross 0 from below. At this time, we reset the statistic $D_n$ to $0$ and go to Step $1)$.
    \end{enumerate}

\item Stopping Rule: Stop at the stopping time
$$ \tau_{\textup{\fontsize{6 pt}{10 pt}\selectfont DE$m$E}} = \inf \{ n \geq 1: D_n > A \}. $$
\end{enumerate}    
\end{algorithm}

We have the following theorem, which establishes the asymptotic optimality of the DE$m$E-CUSUM Algorithm.

\begin{theorem}
\label{thm:deme}

For the \textup{DE$m$E-CUSUM} algorithm, let
$ 0 < D(f_0^j \mathrel{\Vert} f_1^j) < D(f_0^i \mathrel{\Vert} f_1^i) < \infty $
for $i > j$ where $i, j = 1, 2, \dots, m$.

\begin{enumerate}
\item For any fixed $a_i > 0$ ($i = 1, \dots, m$), $N_j > 0$ ($j = 0, \dots, m - 1$), and $\mu > 0$, with $A = \log \gamma$, $\textup{ARLFA} (\Psi_{\textup{DE$m$E}}) \geq \gamma$.

\item For any $A$, we can choose the parameters such that $\textup{POR}_i (\Psi_{\textup{DE$m$E}} (A, a_1, a_2, \dots, a_m, N_0, N_1, \dots, N_{m - 1}, \mu)) \leq \beta_i$. If $\ds \sum_{i = 1}^m \beta_i < 1$, then the data-efficiency constraint is also satisfied.

\item For fixed values of $a_i > 0$ ($i = 1, 2, \dots, m$) and $N_j > 0$ ($j = 0, 1, \dots, m - 1$), and for each $A$,
$$ \textup{WADD} (\tau_{\textup{\fontsize{6 pt}{10 pt}\selectfont DE$m$E}}) \leq \mathsf{E}_1 [\tau_{\textup{\fontsize{6 pt}{10 pt}\selectfont DE$m$E}}] + \sum_{r = 0}^{m - 1} \prod_{k = (m - 1) - r}^{m - 1} N_k. $$
Consequently,
$$ \textup{WADD} (\Psi_{\textup{DE$m$E}}) \leq \textup{WADD} (\tau_{\textup{\fontsize{6 pt}{10 pt}\selectfont C}}^m) + K_m, $$
where $\tau_{\textup{\fontsize{6 pt}{10 pt}\selectfont C}}^m$ is the \textup{CUSUM} algorithm for the experiment $m$ with highest quality and $K_m$ is a constant that is not a function of threshold $A$. Hence, the \textup{DE$m$E-CUSUM} algorithm is asymptotically optimal.
\end{enumerate}

\end{theorem}


\section{Conclusion}
\label{sec:conclusion}

The quickest change detection (QCD) problem typically involves observing a stochastic process using a single experiment. In this paper, we extended the classical problem by considering the case where the observer chooses between multiple experiments with different information quality to minimize WADD subject to ARLFA and POR constraints.

We proposed the Two-Experiment Cumulative Sum (2E-CUSUM) algorithm to solve the two-experiment QCD problem by starting with the higher-quality experiment in order to detect a change quickly in case it occurs early. In this algorithm, we do not limit the number of observations of the higher-quality experiment, but the lower-quality one is limited. This design allows the decision-maker to attain any POR level desired. We also defined the truncated 2E-CUSUM, where we limit the number of observations of the higher-quality experiment as well.

We then explored the three-experiment case and developed the Three-Experiment CUSUM (3E-CUSUM) algorithm, which is an extension of the 2E-CUSUM algorithm. The 3E-CUSUM algorithm begins by performing the highest-quality experiment. Once the statistic falls below 0, the undershoot is scaled, and this scaled value becomes the starting level for the truncated 2E-CUSUM algorithm for the mid- and lowest-quality experiments. When the statistic crosses 0 from below, the statistic is reset to 0, and the highest-quality experiment is performed again. Similarly, the Four-Experiment CUSUM (4E-CUSUM) algorithm was defined in terms of the truncated 3E-CUSUM algorithm, and this was easily be extended to any number $m$ of experiments.

After understanding the issues and gaining insights from two- and three-experiment designs, we then formulated the problem and proposed the $m$E-CUSUM algorithm, an algorithm which solves the $m$-experiment design case.

Data-efficiency, wherein we incorporated periods when no experiment is performed, was also studied for the two-experiment case. An algorithm for the data-efficient two-experiment scenario was easily developed, motivated by the structure of the three-experiment design: the 3E-CUSUM algorithm was modified to develop the Data-Efficient 2E-CUSUM (DE2E-CUSUM) algorithm, which incorporates a deterministic updating of the statistic when no experiment is performed. We extended this to develop the Data-Efficient 3E-CUSUM (DE3E-CUSUM) algorithm, which starts with the highest-quality experiment. Once the statistic falls below 0, the undershoot is scaled, and this scaled value becomes the starting point for the truncated DE2E-CUSUM algorithm. Once the statistic crosses 0 from below, the statistic is reset to 0, and the highest-quality experiment is performed again. We followed this pattern and generalized the algorithm for data efficiency for the case of $m$ experiments.




%



\section*{Acknowledgment}

This work was supported in part by the U.S. National Science Foundation under grant 2427909.


\ifCLASSOPTIONcaptionsoff
  \newpage
\fi



%

\bibliographystyle{IEEEtran}
\bibliography{QCDPaper}

\begin{thebibliography}{10}
\providecommand{\url}[1]{#1}
\csname url@samestyle\endcsname
\providecommand{\newblock}{\relax}
\providecommand{\bibinfo}[2]{#2}
\providecommand{\BIBentrySTDinterwordspacing}{\spaceskip=0pt\relax}
\providecommand{\BIBentryALTinterwordstretchfactor}{4}
\providecommand{\BIBentryALTinterwordspacing}{\spaceskip=\fontdimen2\font plus
\BIBentryALTinterwordstretchfactor\fontdimen3\font minus \fontdimen4\font\relax}
\providecommand{\BIBforeignlanguage}[2]{{%
\expandafter\ifx\csname l@#1\endcsname\relax
\typeout{** WARNING: IEEEtran.bst: No hyphenation pattern has been}%
\typeout{** loaded for the language `#1'. Using the pattern for}%
\typeout{** the default language instead.}%
\else
\language=\csname l@#1\endcsname
\fi
#2}}
\providecommand{\BIBdecl}{\relax}
\BIBdecl

\bibitem{veeravalli2013}
V.~V. Veeravalli and T.~Banerjee, ``Quickest change detection,'' in \emph{Academic Press Library in Signal Processing: Array and Statistical Signal Processing}.\hskip 1em plus 0.5em minus 0.4em\relax Oxford, U.K.: Acad. Press, 2013, pp. 209--255.

\bibitem{poor2009}
H.~V. Poor and O.~Hadjiliadis, \emph{Quickest detection}.\hskip 1em plus 0.5em minus 0.4em\relax Cambridge, U.K.: Cambridge Univ. Press, 2009.

\bibitem{tartakovsky2014}
I.~V.~N. A.~G.~Tartakovsky and M.~Basseville, \emph{Sequential analysis: Hypothesis testing and change-point detection}.\hskip 1em plus 0.5em minus 0.4em\relax Boca Raton, FL, USA: CRC Press, 2014.

\bibitem{tartakovsky2019}
A.~G. Tartakovsky, \emph{Sequential change detection and hypothesis testing: General non-i.i.d. stochastic models and asymptotically optimal rules}, 1st~ed.\hskip 1em plus 0.5em minus 0.4em\relax Chapman and Hall/CRC, 2019.

\bibitem{shewhart1925}
W.~A. Shewhart, ``The application of statistics as an aid in maintaining quality of a manufactured product,'' \emph{J. Am. Stat. Assoc.}, vol.~20, no. 152, pp. 546--548, Dec. 1925.

\bibitem{page1954}
E.~S. Page, ``Continuous inspection schemes,'' \emph{Biometrika}, vol.~41, no. 1--2, pp. 100--115, Jun. 1954.

\bibitem{shiryaev1963}
A.~N. Shiryaev, ``On optimum methods in quickest detection problems,'' \emph{Theory Probab. Appl.}, vol.~8, no.~1, pp. 22--46, Jan. 1963.

\bibitem{lorden1971}
G.~Lorden, ``Procedures for reacting to a change in distribution,'' \emph{Ann. Math. Statist.}, vol.~42, no.~6, pp. 1897--1908, Dec. 1971.

\bibitem{moustakides1986optimal}
G.~V. Moustakides, ``Optimal stopping times for detecting changes in distributions,'' \emph{the Annals of Statistics}, vol.~14, no.~4, pp. 1379--1387, 1986.

\bibitem{lai1998information}
T.~L. Lai, ``Information bounds and quick detection of parameter changes in stochastic systems,'' \emph{IEEE Transactions on Information Theory}, vol.~44, no.~7, pp. 2917--2929, 1998.

\bibitem{tartakovsky2005general}
A.~G. Tartakovsky and V.~V. Veeravalli, ``General asymptotic bayesian theory of quickest change detection,'' \emph{Theory of Probability \& Its Applications}, vol.~49, no.~3, pp. 458--497, 2005.

\bibitem{krishnamurthy2016}
V.~Krishnamurthy, \emph{Partially observed markov decision processes: From filtering to controlled sensing}.\hskip 1em plus 0.5em minus 0.4em\relax Cambridge, U.K.: Cambridge Univ. Press, 2016.

\bibitem{chernoff1959}
H.~Chernoff, ``Sequential design of experiments,'' \emph{Ann. Math. Statist.}, vol.~30, no.~3, pp. 755--770, Sep. 1959.

\bibitem{bessler1960-2}
S.~A. Bessler, ``Theory and applications of the sequential design of experiments, k-actions and infinitely many experiments, part ii applications,'' Dept. Statist., Standford Univ., Rep.~56, Apr. 1960.

\bibitem{albert1961}
A.~E. Albert, ``The sequential design of experiments for infinitely many states of nature,'' \emph{Ann. Math. Statist.}, vol.~32, no.~3, pp. 774--799, Sep. 1961.

\bibitem{kiefer1963}
J.~Kiefer and J.~Sacks, ``Asymptotically optimum sequential inference and design,'' \emph{Ann. Math. Statist.}, vol.~34, no.~3, pp. 705--750, Sep. 1963.

\bibitem{lalley1986}
S.~P. Lalley and G.~Lorden, ``A control problem arising in the sequential design of experiments,'' \emph{Ann. Probab.}, vol.~14, no.~1, pp. 136--172, Jan. 1986.

\bibitem{keener1984}
R.~Keener, ``Second order efficiency in the sequential design of experiments,'' \emph{Ann. Statist.}, vol.~12, no.~2, pp. 510--532, Jun. 1984.

\bibitem{nitinawarat2013}
G.~K.~A. S.~Nitinawarat and V.~V. Veeravalli, ``Controlled sensing for multihypothesis testing,'' \emph{IEEE Trans. Aut. Control}, vol.~58, no.~10, pp. 2451--2464, Oct. 2013.

\bibitem{naghshvar2013}
M.~Naghshvar and T.~Javidi, ``Active sequential hypothesis testing,'' \emph{Ann. Statist.}, vol.~41, no.~6, pp. 2703--2738, Dec. 2013.

\bibitem{nitinawarat2015}
S.~Nitinawarat and V.~Veeravalli, ``Controlled sensing for sequential multihypothesis testing with controlled markovian observations and non-uniform control cost,'' \emph{Sequential Anal.}, vol.~34, no.~1, pp. 1--24, Feb. 2015.

\bibitem{deshmukh2021}
V.~V.~V. A.~Deshmukh and S.~Bhashyam, ``Sequential controlled sensing for composite multihypothesis testing,'' \emph{Sequential Anal.}, vol.~40, no.~2, pp. 1--38, Mar. 2021.

\bibitem{gurevich2019}
K.~C. A.~Gurevich and Q.~Zhao, ``Sequential anomaly detection under a nonlinear system cost,'' \emph{IEEE Trans. Signal Process.}, vol.~67, no.~14, pp. 3689--3703, Jul. Jul. 2019.

\bibitem{hemo2020}
K.~C. B.~Hemo, T.~Gafni and Q.~Zhao, ``Searching for anomalies over composite hypotheses,'' \emph{IEEE Trans. Signal Process.}, vol.~68, no.~2, pp. 1181--1196, Feb. 2020.

\bibitem{vaidhiyan2018}
N.~K. Vaidhiyan and R.~Sundaresan, ``Learning to detect an oddball target,'' \emph{IEEE Trans. Inf. Theory}, vol.~64, no.~2, pp. 831--852, Feb. 2018.

\bibitem{tsopelakos2019}
G.~F. A.~Tsopelakos and V.~V. Veeravalli, ``Sequential anomaly detection with observation control,'' in \emph{Proc. IEEE Int. Symp. Inf. Theory (ISIT)}, 2019, pp. 2389--2393.

\bibitem{tsopelakos2020}
A.~Tsopelakos and G.~Fellouris, ``Sequential anomaly detection with observation control under a generalized error metric,'' in \emph{Proc. IEEE Int. Symp. Inf. Theory (ISIT)}, 2020, pp. 1165--1170.

\bibitem{tsopelakos2023}
------, ``Sequential anomaly detection with sampling constraints,'' \emph{IEEE Trans. Inf. Theory}, vol.~69, no.~12, pp. 8126--8146, Dec. 2023.

\bibitem{veeravalli2024}
G.~F. V.~V.~Veeravalli and G.~V. Moustakides, ``Quickest change detection with controlled sensing,'' \emph{IEEE J. Sel. Areas Inf. Theory}, vol.~5, pp. 1--11, Feb. 2024.

\bibitem{gopalan2021}
B.~L. A.~Gopalan and V.~Saligrama, ``Bandit quickest changepoint detection,'' \emph{Proc. Adv. Neural Inf. Process. Syst.}, pp. 1--10, Nov. 2021.

\bibitem{dragalin1996}
V.~Dragalin, ``A simple and effective scanning rule for a multi-channel system,'' \emph{Metrika}, vol.~43, no.~1, pp. 165--182, Dec. 1996.

\bibitem{zhang2022}
W.~Zhang and Y.~Mei, ``Bandit change-point detection for real-time monitoring high-dimensional data under sampling control,'' \emph{Technometrics}, vol.~65, no.~1, pp. 33--43, Apr. 2022.

\bibitem{xu2021}
Y.~M. Q.~Xu and G.~V. Moustakides, ``Optimum multistream sequential change-point detection with sampling control,'' \emph{IEEE Trans. Inf. Theory}, vol.~67, no.~11, pp. 7627--7636, Nov. 2021.

\bibitem{chaudhuri2021}
G.~F. A.~Chaudhuri and A.~Tajer, ``Sequential change detection of a correlation structure under a sampling constraint,'' in \emph{Proc. IEEE Int. Symp. Inf. Theory (ISIT)}, 2021, pp. 605--610.

\bibitem{xu2023}
Q.~Xu and Y.~Mei, ``Asymptotic optimality theory for active quickest detection with unknown postchange parameters,'' \emph{Sequential Anal.}, vol.~42, no.~2, pp. 150--181, May 2023.

\bibitem{banerjee2013}
T.~Banerjee and V.~V. Veeravalli, ``Data-efficient quickest change detection in minimax settings,'' \emph{IEEE Trans. Inf. Theory}, vol.~59, no.~10, pp. 6917--6931, Oct. 2013.

\bibitem{banerjee2012-2}
------, ``Energy-efficient quickest change detection in sensor networks,'' in \emph{Proc. IEEE Statist. Signal Process. Workshop}, 2012, pp. 636--639.

\bibitem{banerjee2015}
------, ``Data-efficient minimax quickest change detection with composite post-change distribution,'' \emph{IEEE Trans. Inf. Theory}, vol.~61, no.~9, pp. 5172--5184, Sep. 2015.

\bibitem{banerjee2013-2}
------, ``Data-efficient quickest change detection in distributed and multichannel systems,'' Presented at the IEEE Int. Conf. Acoust., Speech, Signal Process., May 2013.

\bibitem{banerjee2013-3}
V.~V.~V. T.~Banerjee and A.~Tartakovsky, ``Decentralized data-efficient quickest change detection,'' Presented at the IEEE Int. Symp. Inf. Theory, Jul. 2013.

\bibitem{banerjee2015-2}
T.~Banerjee and V.~V. Veeravalli, ``Data-efficient minimax quickest change detection in a decentralized system,'' \emph{Sequential Anal.}, vol.~34, no.~2, pp. 148--170, Feb. 2015.

\bibitem{banerjee2015-3}
------, ``Data-efficient quickest change detection in sensor networks,'' \emph{IEEE Trans. Signal Process.}, vol.~63, no.~14, pp. 3727--3735, Jul. 2015.

\bibitem{ren2017}
K.~H.~J. X.~Ren and L.~Shi, ``Quickest change detection with observation scheduling,'' \emph{IEEE Trans. Autom. Control}, vol.~62, no.~6, pp. 2635--2647, Jun. 2017.

\bibitem{pollak1985}
M.~Pollak, ``Optimal detection of a change in distribution,'' \emph{Ann. Statist.}, vol.~13, no.~1, pp. 206--227, Mar. 1985.

\end{thebibliography}




%








\end{document}